\documentclass[journal]{IEEEtran}

\usepackage{algorithm}
\usepackage[noend]{algpseudocode}
\usepackage[pdftex]{graphicx} 
\graphicspath{{./img/}}
\usepackage{amsmath} 
\usepackage{amssymb} 
\usepackage{amsthm} 
\usepackage{url}
\usepackage{subfig}
\usepackage{stmaryrd}
\usepackage[dvipsnames]{xcolor}
\usepackage{colortbl}
\usepackage{tikz}
\usepackage{siunitx}
\usepackage{pgfplots}
\usepackage[hidelinks]{hyperref}
\usetikzlibrary{plotmarks}
\usepackage{caption,xcolor}

\newtheorem{theorem}{Theorem}[section]

\usepackage{thmtools}

\declaretheoremstyle[
  bodyfont=\normalfont\itshape,
  headformat=\NAME\NUMBER\NOTE  
]{nospacetheorem}
\declaretheorem[style=nospacetheorem,name=A]{hypothesis}

\newtheorem{proposition}[theorem]{Proposition}
\newtheorem{lemma}[theorem]{Lemma}

\def\b0{\boldsymbol{0}}
\def\by{\boldsymbol{y}}
\def\bs{\boldsymbol{s}}
\def\ba{\boldsymbol{a}}
\def\bz{\boldsymbol{z}}
\def\bg{\boldsymbol{g}}
\def\bt{\boldsymbol{t}}
\def\bd{\boldsymbol{d}}
\def\bD{\boldsymbol{D}}

\def\bSig{\boldsymbol{\Sigma}}

\def\bI{\boldsymbol{I}}
\def\bx{\boldsymbol{x}}

\def\beps{\boldsymbol{\epsilon}}

\begin{document}
%
\title{Robust control of varying weak hyperspectral target detection with sparse non-negative representation}

\author{Raphael~Bacher,
        Celine~Meillier,
        Florent~Chatelain
        and~Olivier~Michel
\thanks{The authors are with the Image and Signal Department, GIPSA-Lab,
Grenoble Institute of Technology, Saint Martin d'Heres 38400, France (e-mail:
raphael.bacher@gipsa-lab.grenoble-inp.fr).}}

\maketitle

\begin{abstract}
In this study, a multiple-comparison approach is developed for detecting faint hyperspectral sources.
The detection method relies on a sparse and non-negative representation on a highly coherent dictionary to 
track a spatially varying source. A robust control of the detection errors is ensured by learning the 
test statistic distributions on the data. The resulting control is based on the false discovery rate, 
to take into account the large number of  pixels to be tested. This method is applied to data recently recorded by the 
three-dimensional spectrograph Multi-Unit Spectrograph Explorer. 

\end{abstract}

\IEEEpeerreviewmaketitle

\section{Introduction}

\IEEEPARstart{W}{ith} the constant development of new imaging devices, the exploration of massive multi-modalities datasets has become an important field of study, with many challenges to face. In particular, the present study is motivated by the need to detect faint emission line features in massive hyperspectral data produced by the
recent three-dimensional (3D) spectrograph Multi-Unit Spectrograph Explorer (MUSE) instrument \cite{bacon2010muse}. The targeted emission lines are markers of spatially extended structures (or 'halos', surrounding galaxies) that can exhibit spectral variability (spectral shifts). 
Furthermore, the presence of a large number of nuisance sources with high dynamics in the background makes the estimation of the background statistics particularly challenging. Searching faint signals in massive datasets requires the removal of possibly strong contributions of unwanted objects. Typically, for tracking faint signatures in hyperspectral data with high background dynamics, the spectra baseline must be estimated and removed.
Finally, the size of the data to be explored calls for error controls of mis-detections with a global significance, such as the false discovery rate (FDR)\cite{benjamini1995controlling}, to allow unsupervised detection of the targets.

To summarize, the present study addresses a quite general detection problem whose main features are:
\begin{itemize}
\item highly variable, partially known, weak target signatures;
\item background difficult to model;
\item possible presence of strong contributions from unwanted sources that have to be removed;
\item necessity for robust global mis-detections control.
\end{itemize}

Many methods have been developed in recent years to detect targets in hyperspectral data \cite{manolakis2009there}, all requiring knowledge of the background and/or the target signature.
Among these studies, spectral anomaly detectors can be used when the searched signal is unknown, but rely on a parametric statistical modeling of the background signature \cite{reed1990adaptive} 
Spectral matching detectors, such as adaptive matched filters \cite{manolakis2000comparative}
or adaptive cosine estimators \cite{scharf1996adaptive}, exploit prior knowledge of both the target signature and the background characterization. 
A common feature of these approaches is the lack of global control of false alarm rate.
Furthermore, these detectors were developed in a remote-sensing framework \cite{solomon1985imaging} at rather high signal-to-noise ratio (SNR) and can barely be adapted to the new challenges (low SNR, absence of ground-truth...) raised by new massive hyperspectral data such as those provided by new instruments in astronomy. 
An alternate class of methods relies on sparse representation \cite{chen2011sparse} based on a trained dictionary. These are in fact mostly reconstruction methods, exploited for detection purposes. Up to our knowledge these methods do not allow to calibrate the type I error (false alarm) control of the test.
Another pitfall is that in general no training set are available in the astronomical context.
Finally, recent approaches  \cite{paris2013detection}, \cite{courbot2016detection} try to tackle a problem having the same features as ours. However, the Generalized Likelihood Ratio based solutions proposed in these studies do not allow reliable control of detection error. This control is crucial for massive datasets in general, and for the present MUSE dataflow in particular, and is the core of this study.

Rare events or sparse signals detection problems have received much attention in the recent statistical literature. Multiple-comparison procedures, such as higher criticism  or Bonferroni-type methods, have been proposed and shown to have asymptotic optimal detection properties under sparsity regimes \cite{donoho2004higher,hall2010innovated,arias2011global,donoho2015higher}.
Such methods do not require specification of the signals/events to be detected. Higher criticism procedures can be viewed as adaptive to the unknown sparsity level and power of the signals to be detected. 
These multiple-comparison procedures can be applied on the dictionary-based representations of the signals. Overcomplete and/or coherent dictionaries are prone to provide sparse representation matched to the application at hand, and were shown to greatly improve the detection power\cite{donoho2004higher}. Again no global error control is performed by these methods.

The purpose of the present paper is thus threefold:
\begin{itemize}
\item derive a detection algorithm that benefits from the detection power of the sparse representation based approaches;
\item propose a method that requires very weak assumptions on the background or target statistical properties;
\item operate a test procedure that allow a control of the global false alarm rate (FDR).
\end{itemize}

The proposed method is built upon a spectral matching approach over a highly coherent dictionary of target spectra, to take the variability into account. 
It elaborates on the max-test study presented in \cite{arias2011global} for the detection of rare events. As such, it is built on a sparse representation whose purpose is to allow the formulation of a detection test, that do not need any signal reconstruction. To insure a calibration of the test procedure (FDR wise) that is robust to background misspecification, a new simple procedure is proposed, that mainly requires symmetry of the noise distribution.
Note that exploiting this symmetry of the noise versus the positivity of sources in astronomical context was also developed in \cite{serra2012using} but without providing a global control of the errors nor the formalization of a varying target matched over a dictionary of spectral shapes.

The paper is organized as follows.
Section \ref{sec:method} defines the core of the proposed detection approach. The application-oriented design of the dictionary, the data preprocessing step, and the results on the real MUSE data are described 
in section \ref{sec:appli}.
Some conclusions and perspectives are drawn up in section \ref{sec:conclu}.

\subsection*{Notations}
In the following, a 'pixel' refers indifferently to a position in the MUSE spatial grid and to the associated spectrum.  
A pixel or its associated spectrum vector is represented by bold letters e.g. $\bx$, and bold capital letters refer to matrices. 

\section{Detection method}
\label{sec:method}

\subsection{Testing problem}

We address now the detection of a 
signal $\bx$, from noisy data $\by \in \mathbb{R}^l$.
Let $\mathcal{H}_{0}$ and $\mathcal{H}_{1}$ be the hypotheses denoting, respectively, the absence or presence of the 
source contribution $\bx$. The testing problem is:
\begin{align}
\left\{
\begin{array}{ll}
   \mathcal{H}_{0} : \by=\beps,\\ 
    \mathcal{H}_{1} : \by=\bx+\beps,
\end{array}
\right.
\label{eq:testingPb0}
\end{align}
where $\beps \in \mathbb{R}^l$ is a noise vector, centered and independent of $\bx$, for which the distribution is not known. 

When $\bx$ is not fully specified, a classical approach for \eqref{eq:testingPb0} consists of modeling $\bx$ as a sparse superposition
of reference signals taken from a massively overcomplete dictionary $\bD$, 
see for instance \cite{mallat1993matching}, or 
\cite{huang2006sparse} for classification tasks.
The reference signals, or {\em atoms}, correspond to the column  
vector $\bd_j \in \mathbb{R}^l$, for $1 \le j \le m$, of $\bD \in \mathbb{R}^{l\times m}$,
where $m$ is the total number of references. 
These atoms are usually scaled to be $\ell_2$-normalized: $|| \bd_j ||_2=1$  for $1 \le j \le m$.

Moreover, in the present context, the signal of interest 
$\bx$ is generally assumed to be non-negative. Thus, to enforce a non-negative decomposition,
the atoms $\bd_j$, for $1 \le j \le m$ are assumed to be non-negative. It should be noted that, unlike most of the sparse representation techniques in the literature, we do not seek to build an optimal dictionary for reconstruction/estimation but for the design of a good detection test. The dictionary construction, based on physical priors and specific to the application at hand, will be addressed in section \ref{sec:dico} in the framework of MUSE application\footnote{Note that in this application the non-negativity constraints can be relaxed. 
The target $\bx$ and the atoms $\bd_j$ can have negative or positive contributions, as long as their 
dot product $\bx^T\bd_j$ are non-negative for $1 \le j \le m$.}.
Under the non-negativity and sparsity assumptions, the target signal can be expressed as
\begin{align}
 \begin{split}
 & \bx  \approx  a_{i_1} \bd_{i_1}+\ldots + a_{i_k} \bd_{i_k},\\
 & \textrm{s.t.  } a_{i_j} > 0, \quad \textrm{for } 1\le j \le k,
 \end{split}
 \label{eq:sparseDecomp}
\end{align}
where $k \ll m$.
Based on this representation, the detection problem reduces to a multiple comparison 
procedure with one-sided tests:
\begin{align}
 \begin{split}
\left\{
\begin{array}{ll}
   \mathcal{H}_{0} :  a_1=a_2= \ldots=a_m=0,\\ 
    \mathcal{H}_{1} : \textrm{at least one } a_i>0,
\end{array}
\right.
 \end{split}
 \label{eq:testingPb}
\end{align}

\subsection{Test statistic}
\label{se:testingPb}
We now search a test statistic adapted to the detection problem \eqref{eq:testingPb} obtained by sparse representation. Let us first consider the statistic for a single atom.
Let $S(\by,\bd)$ be a measure of similarity between the observed data $\by \in \mathbb{R}^l$ and a normalized reference vector $\bd \in \mathbb{R}^l$. 
Popular examples of similarity measures include the matched filter statistic 
\begin{align}
S(\by,\bd) \equiv \left\langle \frac{\bd}{||\bd||}, \by \right \rangle= \bd^T\by,
\label{eq:MF}
\end{align}
or the spectral angular distance (SAD)  
\begin{align}
S(\by,\bd) \equiv \frac{\langle \bd, \by\rangle}{||\bd|| ||\by||}=\frac{\bd^T\by}{||\by||},
\label{eq:SAD}
\end{align}
which is a  classical distance in hyperspectral analysis \cite{schowengerdt2006remote}. For a given signal amplitude $a=||\by||>0$, such similarity measures are maximized when $\by= a \bd$.

Based now on the pairwise similarity measures $S(\by,\bd_j)$ between the observation $\by$ and the dictionary atoms $\bd_j$, for $1\le j \le m$,
a global test for the multiple (with regards to the $m$ atoms) testing problem introduced in \eqref{eq:testingPb} can be derived from a Bonferroni-like correction. Accounting for \eqref{eq:sparseDecomp}, this leads us to 
consider the following one-sided max-test approach:
\begin{align}
T_{\textrm{max}}(\by) \equiv \max\limits_{1\le j \le m} S(\by, \bd_j)  \underset{H_0}{\overset{H_1}{\gtrless}} \eta,
\label{eq:maxtest}
\end{align}
where $\eta$ is a given threshold. The motivations for using this max-test approach are two-fold. First, from a theoretical point of view, with highly sparse signals, the max-test method is asymptotically as efficient as the asymptotically optimal higher criticism method, as demonstrated in \cite{donoho2004higher,arias2011global}. Secondly, in finite sample settings such as the MUSE hyperspectral datasets, max-test has been shown to be relatively efficient \cite{paris2013detection,meillier2015nonparametric}, and empirically more powerful than higher criticism methods \cite{meillier2015detection}.

We now tackle the problem of applying the max-test defined in \eqref{eq:maxtest} to a large number $n$ of data realisations $\{\by_i\}_{1\leq i \leq  n}$. We are again in a multiple-testing context, now with regards to the number of samples $n$. This is in regards to this context that we seek to control the detection errors.
To fix the decision threshold $\eta$ while controlling the type I errors, i.e., the samples under $\mathcal{H}_0$ that will be falsely detected as $\mathcal{H}_1$, 
 the distribution under the null hypothesis $\mathcal{H}_0$ of the  max-test statistic $T_{\textrm{max}}(\by)$  must be known. 
In real applications such as the MUSE data, due to the physical process and preprocessing steps (e.g., interpolation, background subtraction), noise is spatially and spectrally correlated with an unknown complex dependence structure. Thus the distribution of $T_{\textrm{max}}(\beps)$, where $\beps$ is the noise vector introduced in  \eqref{eq:testingPb0},
 cannot be easily modeled as a standard parametric distribution. However in a large scale testing framework, it becomes possible to estimate this distribution, as explained in the next section.

\subsection{Learning the null distribution}
\label{subsec:H0}

Consider now the following assumptions:
\begin{hypothesis}
The noise vector $\beps$ is centered and symmetrically distributed: $\beps$ and $-\beps$ have the same distribution,
\end{hypothesis}
\begin{hypothesis}
The similarity measure $S(\by,\bd)$ used to construct the max-test statistics is an odd function of the observations $\by$, i.e. $S(\by,\bd)=-S(-\by,\bd)$ for any 
sample $\by$ and for any reference vector $\bd$.
\end{hypothesis}
Assumption A1 on the noise is relatively weak and fairly reasonable. This is, for instance, satisfied for any centered elliptical distribution, such as multivariate Gaussian or student distributions. 
Note also that assumption A2 is clearly satisfied for the matched filter  or the SAD  statistics described respectively in \eqref{eq:MF} and \eqref{eq:SAD}. 
A direct consequence of these assumptions is the following key property:
\begin{proposition}
 Based on assumptions A1 and A2,  the max-test statistic $T_{\textrm{max}}(\by)$ and the opposite of the min statistic  $-T_{\textrm{min}}(\by)$, where
 $T_{\textrm{min}}(\by) \equiv \min\limits_j S(\by, \bd_j)$, are identically distributed under the null hypothesis $\mathcal{H}_0$.  
 \label{prop:minstat}
\end{proposition}
\begin{proof}
Under the null hypothesis $\by=\beps$. According to A1, $T_{\textrm{max}}(\beps)$ and  $T_{\textrm{max}}(-\beps)$ are identically distributed.
Moreover, 
\begin{align*}
 T_{\textrm{max}}(-\beps)&= \max\limits_j S(-\beps, \bd_j)=  -\min\limits_j \left\{ -S(-\beps, \bd_j) \right\},\\
 &=-\min \limits_j S(\beps, \bd_j)= -T_{\textrm{min}}(\beps),
\end{align*}
where the first equality on the second line is due to A2. Thus $T_{\textrm{max}}(\beps)$ and -$T_{\textrm{min}}(\beps)$ have the same distribution.
\end{proof}

In a large-scale testing framework (with regards to the number of samples $n$), the max and min statistics $T_{\textrm{max}}(\by)$ and $T_{\textrm{min}}(\by)$ are computed for a large number of observations $\by_i$, for $1\le i \le n$. Let $\pi_0 \in (0,1]$ be the true proportion of observations $\by_i$ distributed according to the null hypothesis, while $\pi_1=1-\pi_0$ is the proportion of observations distributed according to the alternative hypothesis for the testing problem \eqref{eq:testingPb0}. Let $F(t)=\Pr{\left(T_{\textrm{max}}(\by) \le t\right)}$ be the cumulative distribution function of the max statistic $T_{\textrm{max}}(\by) $. This distribution function can be expressed as a two-groups model:
\begin{align*}
 F(t)= \pi_0 F_0 (t)+\pi_1 F_1 (t),
\end{align*}
where $F_0$ and $F_1$ denote the distribution functions under the null and the alternative hypotheses, respectively. Under the non-negativity assumption introduced in
section \ref{se:testingPb}, 
$T_{\textrm{max}}(\by)$  should be stochastically larger under $\mathcal{H}_{1}$ than under $\mathcal{H}_{0}$, i.e., $F_0(t) > F_1(t)$ for any $t\in \mathbb{R}$. 
Let $\mu_0$ be the median of the max test statistics under the null hypothesis\footnote{For the sake of simplicity,  the observations are assumed to obey absolutely continuous distributions. Thus the test statistics $T$ are also continuous, and their median $\mu$ is defined as $\Pr( T \le \mu )=\Pr( T \ge \mu )=\frac{1}{2}$. The extension to discrete statistics is left to the reader.} (referred to as the {\em null median} hereafter), i.e., $F_0 (\mu_0)= \frac{1}{2}$. We now introduce the classical {\em zero assumption}, as termed by Efron a different context \cite[Chap. 6]{efron2012large}. This assumes the existence of a noise-only domain that allows to build a procedure for estimating the null distribution (see {\em Remark 2} hereinafter for a discussion about this assumption).
\begin{hypothesis}[Zero assumption for $F_1$]
 $F_1(t) = 0$ for the region $t \le \mu_0$ where the max statistics are most likely under $\mathcal{H}_{0}$.
\end{hypothesis}
From this assumption, we can now derive the following expression:
\begin{align*}
 F(t) = \pi_0 F_0 (t), \quad \textrm{for } t\le \mu_0.
\end{align*}
In a similar manner,  the survival function $G(t)=\Pr{\left(-T_{\textrm{min}}(\by) >t\right)}$ of the opposite min statistic $-T_{\textrm{min}}(\by)$ reads as
\begin{align*}
 G(t)= \pi_0 G_0 (t)+\pi_1 G_1 (t),
\end{align*}
where $G_0$ and $G_1$ are the survival functions of $-T_{\textrm{min}}(\by)$ under the null and  alternative hypotheses, respectively. 
This comes from the non-negativity assumption 
that $-T_{\textrm{min}}(\by)$ should be stochastically smaller under $\mathcal{H}_{1}$ than under $\mathcal{H}_{0}$, i.e., 
$G_0(t) > G_1(t)$. Note that $\mu_0$ is also the median of the null distribution of $-T_{\textrm{min}}(\by)$ as 
$G_0 (\mu_0)= 1-F_0(\mu_0)= \frac{1}{2}$ according to the proposition \ref{prop:minstat}. This allows us to introduce the following
zero assumption.
\begin{hypothesis}[Zero assumption for $G_1$]
 $G_1(t) = 0$ for the region $t \ge \mu_0$ where the opposite min statistics are most likely under $\mathcal{H}_{0}$.
\end{hypothesis}
\noindent Thus
\begin{align*}
 G(t) = \pi_0 G_0 (t),  \quad \textrm{for } t \ge \mu_0.
\end{align*}
Since $G_0 (t)=1-F_0(t)$  according to the proposition \ref{prop:minstat}, we can finally derive the following expression:
\begin{align}
 \pi_0 F_0 (t) & = \begin{cases}
                    F(t), & \textrm{for } t \le \mu_0,\\
                    \pi_0- G(t), & \textrm{for } t > \mu_0.\\
                  \end{cases}
                  \label{eq:expH0}
\end{align}

The main interest of this expression is that it does not depend on each group distribution function but only on the distribution function of the two-groups model. In particular, assumptions A3 and A4 do not require to fully specify $F_1$, which is unlikely to be known in practice.  Expression 
\eqref{eq:expH0} is therefore robust to alternative miss-specifications. This expression still depends on the theoretical  null median $\mu_0$, the proportion $\pi_0$ of samples under $\mathcal{H}_0$, and the distribution functions $F(t)$ and $G(t)$, which are not known. 
However, when a large number of observations $\by_1,\ldots,\by_n$ are available, these  quantities can be estimated from the empirical distributions of the test statistics. 
Let 
\begin{align*}
 \overline{F}(t) =\frac{\#\left\{T_{\textrm{max}}(\by_i) \le t\right\} }{n}, \quad
 \overline{G}(t) =\frac{\#\left\{-T_{\textrm{min}}(\by_i) > t\right\} }{n},
\end{align*}
be the empirical distribution function of $T_{\textrm{max}}$ and the empirical survival function of $-T_{\textrm{min}}$, respectively. 
\begin{hypothesis}[Weak dependence assumption]
The empirical functions $\overline{F}(t)$ and $\overline{G}(t)$ converge uniformly toward the theoretical distribution functions $F(t)$ and $G(t)$,  respectively:
\begin{align*}
 \sup_{t} |\overline{F}(t) - F(t)| \longrightarrow 0, \quad
 \sup_{t} |\overline{G}(t) - G(t)| \longrightarrow 0,
\end{align*}
almost surely as the number $n$ of observations grows to infinity.
\end{hypothesis}
Statement A5 is verified under weak dependence conditions between the observed samples $\by_1,\ldots,\by_n$. In particular, 
this holds for independent or short-range dependent samples according to the Glivenko-Cantelli theorem. A direct consequence is the pointwise convergence in probability of $\overline{F}(t)$ and $\overline{G}(t)$ toward $F(t)$ and $G(t)$, 
respectively, for any $t\in \mathbb{R}$.

Due to zero assumptions A3 and A4 and proposition \ref{prop:minstat}, $\mu_0$ satisfies $F(\mu_0)=G(\mu_0)=\frac{\pi_0}{2}$. Therefore, based on assumption A5, 
an estimator of the null median $\mu_0$ can be searched as a solution of the following equation for $\mu$:
\begin{align}
 \overline{F}(\mu)=\overline{G}(\mu).
 \label{eq:crossmed}
\end{align}

\begin{lemma}[Empirical null median estimator]
\label{le:med}
Let $t_{(1)}<t_{(2)}<\ldots<t_{(2n)}$ be the ordered values of the statistics belonging to the sample 
$\bt=\left(T_{\textrm{max}}(\by_1),\ldots,T_{\textrm{max}}(\by_n),-T_{\textrm{min}}(\by_1),
\ldots,-T_{\textrm{min}}(\by_n)\right)$.
Let $\widehat{\mu}_0$ be the sample median of $\bt$, which is defined as
\begin{align}
 \widehat{\mu}_0= \frac{t_{(n)} + t_{(n+1)}}{2}.
 \label{eq:mu0est}
\end{align}
Then $\widehat{\mu}_0$ satisfies \eqref{eq:crossmed} and is a consistent estimator of the null median $\mu_0$, under zero assumptions A3 and A4.
\end{lemma}

\begin{proof}
See appendix \ref{app:le}.
\end{proof}

Based on the null median estimator given in lemma \ref{le:med}, we can now obtain the empirical null estimates of $\pi_0$ and $F_0(t)$. 
Let
\begin{align*}
 \bs_0=\left\{  T_{\textrm{max}}(\by_i) \le \widehat{\mu}_0 \right\},
\end{align*}
be the sample set of the 
max-test statistics truncated on $(-\infty,\widehat{\mu}_0]$, the elements of which are denoted as $s_{0,i}$, for
$1 \le i \le n_0$, and where $n_0=|\bs_0|$. 
Similarly,  
\begin{align*}
\bg_0=\left\{  -T_{\textrm{min}}(\by_i) > \widehat{\mu}_0 \right\},
\end{align*}
denotes the set of the 
opposite min statistics truncated on $(\widehat{\mu}_0,+\infty)$,  the elements of which are denoted as $g_{0,i}$ for
$1 \le i \le n_0$ (according to lemma \ref{le:med}, these two sets are of equal size).

\begin{proposition}[Empirical estimators under $\mathcal{H}_0$]
\label{prop:nullest}
Under assumptions A3 and A4,
\begin{align}
 \widehat{\pi}_0&= \min\left\{ 2 {n_0}/{n},1 \right\},
 \label{eq:pi0}
 \end{align}
is a consistent estimator of the null proportion $\pi_0$, and
\begin{align}
 \widehat{F}_0(t)&= \frac{ \# \left\{ s_{0,i} \le t \right\} \, + \,  
 \# \left\{ g_{0,i} \le t \right\}}{2 n_0} 
 \label{eq:F0empest}
\end{align} 
is a pointwise consistent estimator of the null distribution $F_0(t)$,
for $t\in \mathbb{R}$.
\end{proposition}
\begin{proof}
See appendix \ref{app:prop}.
\end{proof}

{\em Remark 1:} the dependence structure across a set of observations  $\by_1, \ldots,\by_n$ with $\by_i \in \mathbb{R}^l$ is not required to specify the empirical estimators $\widehat{\pi}_0$ and $\widehat{F}_0$ given in proposition \ref{prop:nullest}. These non-parametric estimates  rely essentially on the noise symmetry assumption A1, which is very weak. As a consequence, these 
estimators are robust to miss-specifications that are prone to occur with parametric assumptions.

{\em Remark 2:} Zero assumptions A3 and A4 provide an idealized mathematical framework for which the empirical estimators are shown to be consistent. However, these assumptions are
unlikely to be satisfied in practice. As a consequence, Eq. \eqref{eq:expH0} is an approximation. Note, however, that the closer $\pi_0$ is to one, the more accurate the approximation is. This is the case in many large-scale testing problems. where the null proportion $\pi_0$ is usually close to one. The approximation gains also in accuracy the more $F_0 (t)$  (resp. $G_0(t)$) dominates $F_1 (t)$ (resp. $G_1(t)$) for $t\le \mu_0$ (resp. for $t\ge \mu_0$). 
Moreover, if a few observations distributed according to the alternative distribution belong to the regions where they are assumed to be absent, then $\widehat{\pi}_0$ tends to be biased upward, and $\widehat{F}_0(t)$ tends to be slightly biased toward the alternative distribution $F_1(t)$. It is of note that from a statistical testing perspective, this slight bias goes in the good way. Indeed, a detection procedure based on $\widehat{F}_0(t)$ (and possibly  $\widehat{\pi}_0$) 
then becomes more conservative as $p$-values are slightly biased upward. This results in a small loss of power but the control of type I errors is still (asymptotically) guaranteed.

Figure \ref{fig:H0b} shows the empirical density functions associated with 
the max-test statistics $T_{\textrm{max}}(\by_i)$ and the opposite min statistics 
$-T_{\textrm{min}}(\by_i)$ for synthetic data with a testing framework 
that mimics the MUSE one.
We can see that the max density has a heavier right tail than the opposite min one. This is due to the 
contribution of the $\mathcal{H}_1$ samples, while the opposite min density right tail is (approximately) distributed according to the theoretical 
null density due to (approximation) A4. By symmetry, the opposite min density has a heavier left tail than the max density one.

To appreciate the accuracy of the empirical estimators given in proposition 
\ref{prop:nullest}, the empirical null-density function that is obtained from the right-truncated sample $\bs_0$ and the
left-truncated sample $\bg_0$ that are used to construct $\widehat{F}_0(t)$, is depicted in Figure \ref{fig:H0c}. 
Here, the null proportion estimate is larger than the theoretical value: $\widehat{\pi}_0=0.89$ and $\pi_0=0.81$.
Nevertheless, the empirical null-density function is very close to the theoretical one. 
This is confirmed by the quantile-quantile plot between $\widehat{F}_0$ and the theoretical distribution $F_0$
shown in Figure \ref{fig:H0d}. In particular, this remains true in the distributions tails, where the accuracy of the 
quantile estimates is crucial to the robustness of the test at low control levels.

\begin{figure}[hbtp!]
\centering
\includegraphics[width=0.8\linewidth]{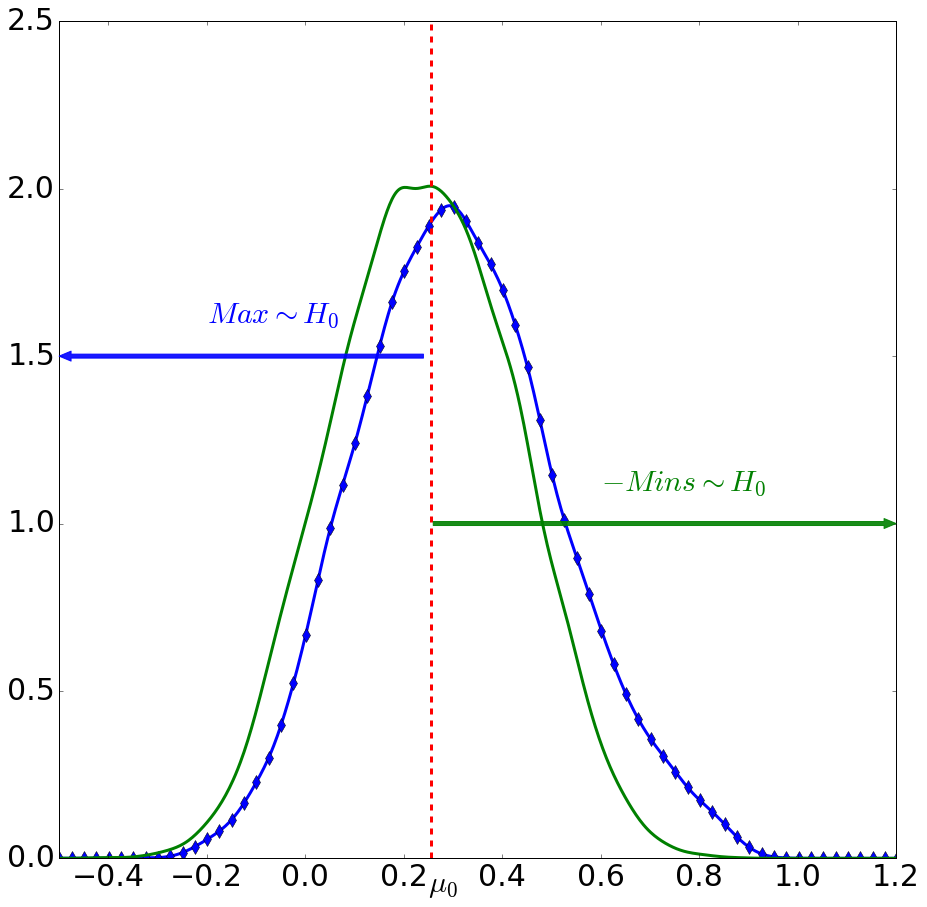}
\caption{Empirical density functions of the max-test statistics $T_{\textrm{max}}(\by_i)$ (blue line with $\begingroup\color{blue}\blacklozenge\endgroup$ markers)
and of the opposites of min statistics $-T_{\textrm{min}}(\by_i)$ (green line) for $n=2500$ independent samples $\by_1,\ldots,\by_n \in \mathbb{R}^l$ generated from the observation model $\eqref{eq:testingPb0}$ with $l=30$. Noise vector $\beps_i$ has i.i.d. entries that follow a student distribution with $\nu=5$ degrees of freedom. The proportion of null hypotheses is $\pi_0 = 0.81$. Samples under $\mathcal{H}_1$ are generated as $\by_i= a_i\bd+\beps_i$,  where $a_i \in [0.1,3]$ and $\bd$ has unit length. The min and max-test statistics are obtained from a dictionary $\bD$ with $m=15$ positively correlated atoms, including $\bd$, and using the SAD similarity measure.
\label{fig:H0b}}
\end{figure}

\begin{figure}[hbtp!]
\centering
\subfloat[Empirical density function of the sample set $\bt_0=\bs_0 \cup \bg_0$, where $\bs_0$ and  $\bg_0$ are the sample sets used to construct
$\widehat{F}_0(t)$ (blue curve) and theoretical density function of $T_{\textrm{max}}(\by)$ under $\mathcal{H}_0$ (green curve with  $\begingroup\color{OliveGreen}\blacklozenge\endgroup$ marker).\label{fig:H0c}]{
\includegraphics[width=0.8\linewidth]{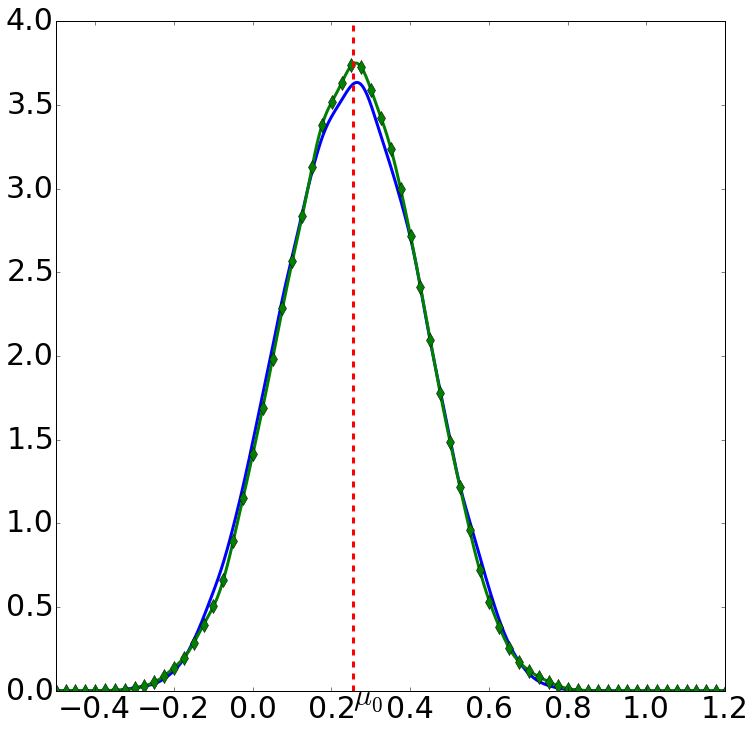}}\\
\subfloat[QQ-plot of $\widehat{F}_0(t)$ against the theoretical distribution $F_0$, $\begingroup\color{OliveGreen}\bullet\endgroup$ marker,
and $y=x$ line (red dashed line). \label{fig:H0d}] {
\includegraphics[width=0.8\linewidth]{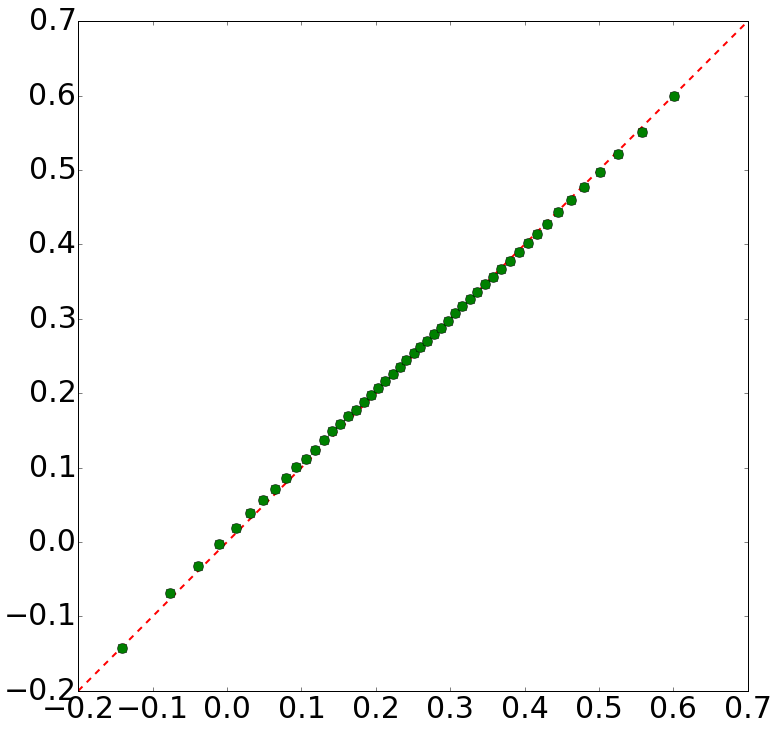}}
\caption{Comparisons between the empirical null distribution  estimator $\widehat{F}_0(t)$ and the theoretical distribution of $T_{\textrm{max}}(\by)$ under $\mathcal{H}_0$. Data are generated using the same setting as in Figure \ref{fig:H0b}. The theoretical distribution function was estimated using $10^5$ Monte-Carlo runs.\label{fig:compareH0}}
\end{figure}

\subsection{Error control}

In multiple testing (around $n=2500$ tested pixels for a $50\times50$ patch in the MUSE context), the classical Type I error control of each individual test might not be appropriate; see e.g., \cite{meillier2015error,genovese2002thresholding}. Indeed the number of wrongly rejected null hypotheses can become relatively important (i.e., even larger than
the number of true detections) due to the high number of tests. To address this kind of issue, a global error control approach, namely the FDR,  was introduced in \cite{benjamini1995controlling}.
The FDR controls the expected proportion of true null hypotheses wrongly rejected, which are referred to as the {\em false discoveries}, among all of the rejected tests:
\begin{align*}
 \textrm{FDR} &= E \left[ \frac{U}{\max{(R,1)}} \right],
\end{align*}
where $R$ is the total number of tests where the null hypothesis is rejected, while $U$ is the number of false discoveries among the $R$ discoveries.
A simple and widely used approach to control this FDR is the Benjamini and Hochberg (BH) procedure that was also developed in \cite{benjamini1995controlling}.
Let $p_i$ be the $p$-value associated with the $i$th test statistics.  
Let $p_{(1)} \le p_{(2)} \le \ldots \le p_{(n)}$ now be the ordered $p$-values, and 
$\mathcal{H}_0^{(1)}, \ldots, \mathcal{H}_0^{(n)}$ denote the null hypotheses for this ordering.
For a preselected control level $0\le q \le 1$, the BH$_q$ procedure rejects $\mathcal{H}_0^{(1)}, \ldots, \mathcal{H}_0^{(\hat{k})}$
where 
\begin{align*}
\hat{k}= \max{\left\{0 \le k \le n : \  p_{(k)} \le q\frac{k}{n}\right\}},
\end{align*}
with $p_{(0)}=0$ by convention.
In our right-sided detection framework, the $i$th {\em empirical} $p$-value, can be derived 
from the empirical null distribution given in proposition \ref{prop:nullest} as
\begin{align}
 p_i= 1-\widehat{F}_0 \left( T_{\textrm{max}}(\by_i) \right), \quad \textrm{for } 1\le i \le n.
 \label{eq:pvalemp}
\end{align}

Then in the case of $n$ independent tests, or under specific positive dependences \cite{benjamini2001control}, the BH$_q$ procedure controls the FDR at a level
$\pi_0 q \leq q$.
Thus, if $\pi_0$ is known, the BH procedure can be applied at the nominal level $\frac{q}{\pi_0}$ to improve its power while controlling FDR at level $q$.
Building on this idea, Storey \cite{storey2004strong} proposed the following {\em modified} estimator of the null proportion  $\pi_0$:
\begin{align*}
  \hat{\pi}_0^{*}(\zeta) = \min\left\{ \frac{ 1+ \#\{ p_{i} >\zeta \}}{(1-\zeta) n}, 1 \right\}, \quad \textrm{for } \zeta \in [0,1),
\end{align*}
where $\zeta$ has to arbitrarily fixed (usually at $\frac{1}{2}$).
Storey showed that under the weak dependence assumption A5, the BH$_{q'}$ procedure at nominal level $q'=q/\hat{\pi}_0^{*}(\zeta)$ asymptotically 
controls the FDR at level $q$. We show in the following that the same strategy can be applied with the empirical null-estimator $\widehat{\pi}_0$.

\begin{proposition}[Storey $\pi_0$ estimator]
\label{prop:storey}
The empirical null estimator $\widehat{\pi}_0$ defined in \eqref{eq:pi0} 
and the Storey estimator $\hat{\pi}_0^{*}(\zeta)$ derived from the empirical
$p$-values defined in \eqref{eq:pvalemp} are equal for any $\zeta= \frac{k}{2n_0}$ with $k\in \{n_0,\ldots,2n_0-1\}$, and are 
asymptotically equivalent for any $\zeta \in [\frac{1}{2},1)$.
\end{proposition}
\begin{proof}
See appendix \ref{app:storey}.
\end{proof}
\noindent This equivalence is not surprising since, like the proposed empirical null estimators,  
Storey estimator is based on a zero assumption (i.e., the $p$-value density function under $\mathcal{H}_1$ is zero on $(\zeta,1]$).

This leads us to consider the following multiple testing procedure described in Alg. \ref{alg:FDR}.
\begin{algorithm}[H]
\caption{FDR-based detection procedure\label{alg:FDR}}
\label{algo:FDR}
\emph{Input:}  nominal FDR level $q$ 
\begin{enumerate}
 \item compute the empirical null estimators $\widehat{\pi}_0$ and $\widehat{F}_0$ as defined in proposition \ref{prop:nullest};
 \item compute the empirical $p$-values according to \eqref{eq:pvalemp};
 \item apply the BH procedure at a nominal control 
level $q / \widehat{\pi}_0$.
\end{enumerate}
\end{algorithm}

Note that it has been shown in \cite{meillier2015error} that for matched filter statistics, with a non-negative template and under Gaussian assumptions, the test statistics obey a positive regression dependence on a subset (PRDS) condition. Therefore, the BH procedure ensures exact FDR control in finite sample settings. Here the problem is more
complex. The test statistics are derived from extreme values that can be correlated. Then 
PRDS is difficult to ensure theoretically, even under Gaussian assumptions on the noise. However, under weak dependence assumption A5, 
an Oracle procedure similar to Alg. \ref{algo:FDR}, but where the $p$-values 
are computed from the theoretical null distribution $F_0$, can be proven to control asymptotically the FDR at level $q$ \cite{storey2004strong}. 
As discussed in the previous subsection, the $p$-value empirical estimates tend to be slightly biased in a 
conservative way. Moreover, if the null distribution can be estimated on a larger sample than the sample to be tested, 
the variance of these estimates can be reduced. This supports the asymptotic control of the proposed procedure.

\subsection{Validation}

Figure \ref{fig:fdr_control} shows the FDR obtained with Alg. \ref{alg:FDR}, on 3D (spatial + spectral) simulated data that are subjected to weak dependence (spatial kernel convolution), for different levels of nominal control $q$ and different signal-to-noise ratio (SNR). The SNR is defined here as $10 \log\frac{A}{n  l \sigma^2}$, where $n$ is the number of pixels (i.e., the number of tests to perform), $l$ is the number of spectral bands (i.e., the dimension of the observations $\by_i$), 
$\sigma^2$ is the marginal variance of the noise, and $A=||\bx||^2$ is the energy of the 3D contribution of the signals to be detected. The experimental set-up is similar to Figure \ref{fig:H0b}, except that a spatial convolution kernel of size 3 by 3 was applied to create local spatial 
correlations. The empirical null estimators defined on section \ref{subsec:H0} were computed from extended cubes of 200 by 200 pixels by 30 wavelengths.

This figure emphasizes that control of the FDR is correctly achieved for the different SNR levels. As expected, due to the zero assumption approximations, Alg. \ref{alg:FDR} is a little more conservative than the Oracle one (based on the true $F_0$ and $\pi_0$) at low SNR, where the alternative distribution is closer to the null one.

\begin{figure}[h]
\centering
\includegraphics[width=0.8\linewidth]{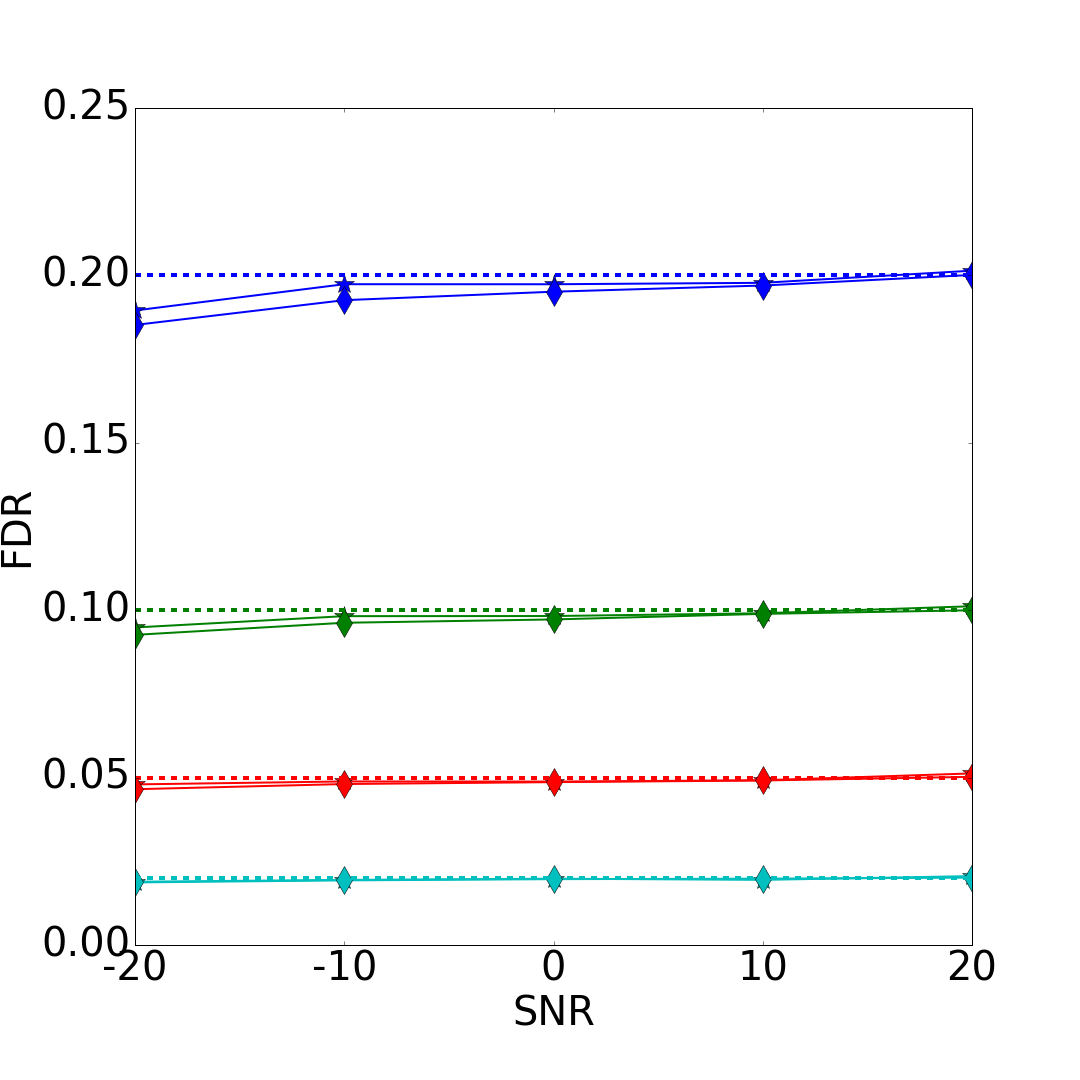}
\caption{Empirical FDR (averaged from 1000 Monte-Carlo runs) versus the SNR under weak dependence, for different levels of FDR control $q=0.02,0.05,0.1,0.2$ (in cyan, red, green, and blue, respectively).
Dashed horizontal lines, nominal control levels $q$;  $\blacklozenge$ curves, FDR for Alg. \ref{alg:FDR} empirical procedure; $\star$ curves, FDR for  Oracle procedure based on the 
true (computed from  $10^5$ Monte-Carlo runs) $F_0$ and $\pi_0$. Data are generated as 3D cubes of $n= 51 \times 51$ pixels by $l=30$ wavelengths. 
}
\label{fig:fdr_control}
\end{figure}

Table \ref{tab:FDR_PFA} shows the advantage of assuring a global control with a detection threshold that adapts to the data. It compares this global control with a pixel-wise control based on the probability of false alarm (PFA). Controlling with a $\eta$ (e.g. 5\%) PFA threshold results in detecting all pixels with $p$-values smaller than $\eta$. Such a PFA control then ensures that in average a fraction $\eta$ of all the tested pixels will be wrongly detected (but says nothing of the proportion of these wrong detections among the detected set).

When confronted to noise-only data, a control procedure with PFA at level 5\% detects 144 spurious pixels, that is the size of a possible source. To insure that no source is falsely detected, we may turn to a more conservative level, e.g. 0.1\%; this results in poor source detection power (around 55\%) whereas a 5\% level led to very good power (82\%) at the price of a large number of false alarms (false discovery proportion $\simeq 43\%$). On the same dataset, a FDR control at 20\% does not lead here to any false detection in the absence of source while maintaining a high detection power (72\%) in presence of a source, thus adapting to the data. Table \ref{tab:FDR_PFA} shows that the false discovery proportion is around 18\% for a nominal FDR control of 20\%. Note that by applying our procedure we can estimate the FDR level that matches a given detection threshold. 
For instance a threshold on the $p$-values associated with a 5\% PFA level yields a FDR estimate around 44\% on data tested here.
Table \ref{tab:FDR_PFA} shows that the false discovery proportion is indeed around 44\% for a threshold corresponding to a 5\% PFA. 

\begin{table}
{

\caption{Comparison between FDR and PFA control on data with and without target. Data was built from noise-only regions in MUSE real data and a synthetic source was added. Number of tests is 2500 (50 by 50 pixels) and source size is 185 pixels. Results were averaged on 5 different regions.}
\begin{tabular}{c|c|c|c}
 & PFA 0.05 & PFA 0.001 & FDR 0.2  \\
\hline
\textit{Noise-only} & & \\
False Detections (pixels) & 144 & 2 & 0 \\
True Detections (pixels) & 0 & 0 & 0 \\
\hline
\textit{Source + noise area} & & & \\
False Detections (pixels) & 117 & 2 & 30 \\
True Detections (pixels) & 153 & 106 & 133  \\
False discovery proportion (\%) & 43.3 & 1.8 & 18.4  \\
\end{tabular}
\label{tab:FDR_PFA}
}
\end{table}

In the next paragraph the ability of the proposed method to control error rate is compared with a generalized likelihood ratio (GLR) approach (inspired by \cite{paris2013detection} and \cite{courbot2016detection}). Noise is supposed centered Gaussian $\beps \sim \mathcal{N}(\b0,\bSig)$ where the covariance matrix $\bSig \in \mathbb{R}^{l\times l}$ is assumed to be diagonal.
We have the following detection test :
$$
\left\{
\begin{array}{ll}
   \mathcal{H}_{0} : \by=\beps,\\ 
    \mathcal{H}_{1} : \by=\bD\ba+\beps, \text{ with } ||\ba||_0=1, \ba\geq\b0
\end{array}
\right.
$$
where $||.||_0$ is the $\ell_0$-pseudo-norm (number of non-zero components) and $\ba\geq\b0$ is the non-negativity constraint on the coefficients.
The GLR test with 1-sparsity constraint yields the following test statistic(\cite{paris2013detection})
$$T_{GLR}(\by)=\frac{\max\limits_{\ba} p(\by|\bD\ba,\mathcal{H}_1)}{p(\by|\mathcal{H}_0)} \text{ s.t. } ||\ba||_0=1, \ba\geq\b0,$$
where $p(\by|\bD\ba,\mathcal{H}_1)$ denotes the probability density function of $\by$ under $\mathcal{H}_1$ and $p(\by|\mathcal{H}_0)$ denotes the probability density function of $\by$ under $\mathcal{H}_0$.
Using the Gaussian assumption it comes that
$$T_{GLR}(\by) = \frac{\bd_{\hat{j}}^T\hat{\bSig}^{-1}\by}{\sqrt{\bd_{\hat{j}}^T\hat{\bSig}^{-1}\bd_{\hat{j}}}}$$
where $\hat{j}$ is the index of the non-zero component of the optimal $\hat{\ba}$ for the GLR statistics, and $\hat{\bSig}$ is estimated from the residuals.
There is no closed-form expression  of the distribution of this statistic since it consists in taking the max of a correlated Gaussian vector. Thus we calibrated this statistic under $\mathcal{H}_0$ (normal centered noise) by Monte-Carlo.

Figures \ref{fig:GLR_norm} and \ref{fig:GLR_stu} illustrate the main advantage of the proposed method : the control is ensured when the noise distribution is symmetrical without further assumptions. The GLR approach has to be calibrated under $\mathcal{H}_0$ distribution so any deviation from the theoretical $\mathcal{H}_0$ results in a loss of control, as illustrated by figure 5 where the noise is drawn from a Student distribution. It can be seen that BH procedure based on the theoritical $\mathcal{H}_0$ GLR statistics do not correctly control FDR. In a first time the effective FDR strongly exceeds the given control level (allowing GLR to be ``more powerful" at a given nominal control level). In a second time it becomes too conservative. This comportment can be explained by the Gaussian fit of the  Student distribution of noise: tails are underestimated (hence the excess in FDR) whereas the bulk is overestimated (hence the loss in power in the second part).
It should be noted that a classical ROC curve of the two methods would show very similar performance (same power for an effective error budget) between the two approaches but would hide the inaccuracy of error control of the GLR approach.
Moreover figure \ref{fig:GLR_norm} shows that when GLR is at its best (adequate model), the proposed method does stays really close in term of power despite its versatility. 

\begin{figure}[hbtp]
\centering
\includegraphics[width=\linewidth]{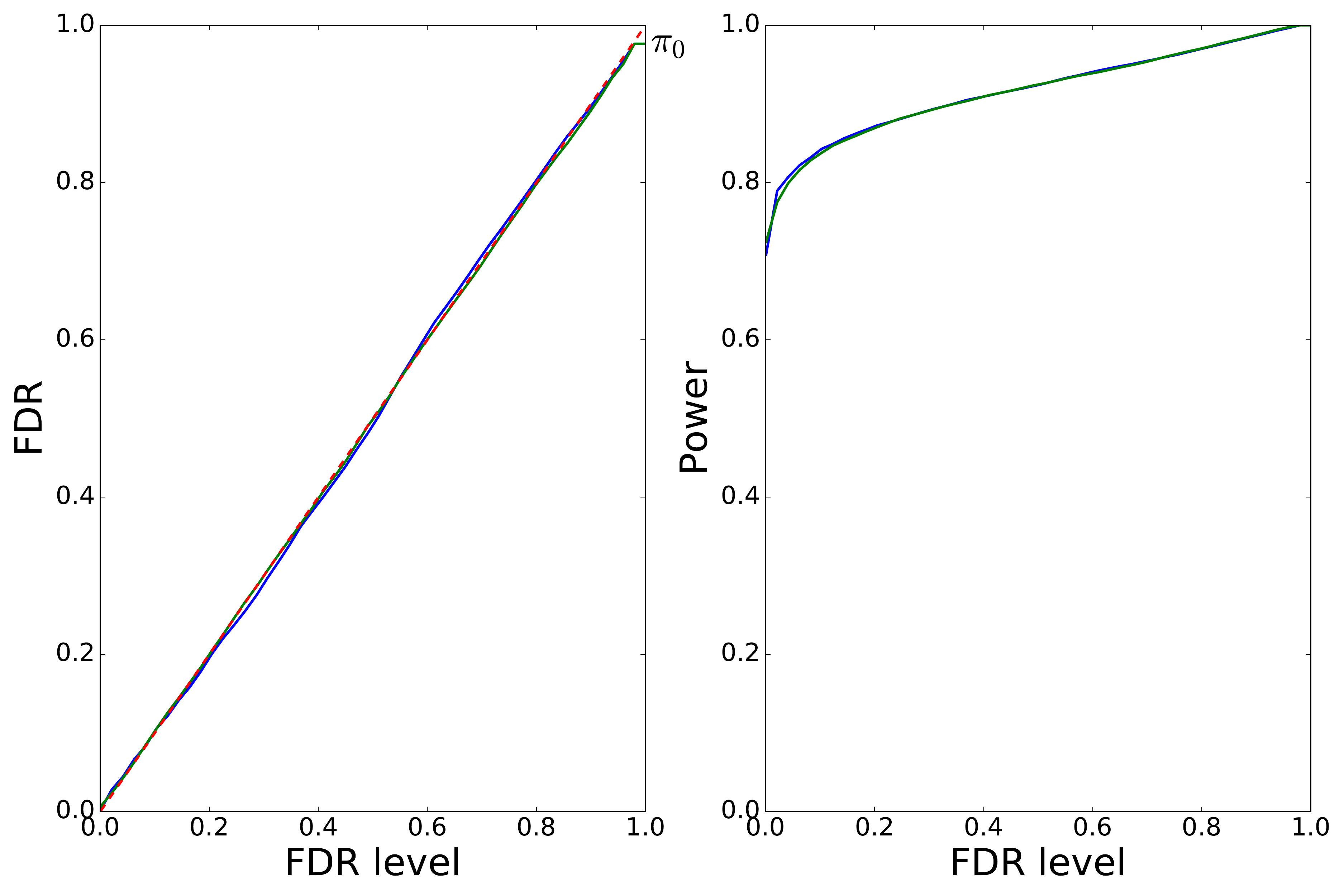}

\caption{True FDR (left) and power (right) versus nominal FDR control level, between GLR and proposed method, on synthesized data with Gaussian noise. The GLR was calibrated by $10^4$ Monte-Carlo runs on normal noise. $\bSig$ was estimated on the data. Results were averaged on 200 runs of simulated data cubes with $\pi_0=0.97$. GLR is in blue, proposed method in green, $y=x$ in dashed red. Matched filter similarity measure (4) was used for the proposed method.}
\label{fig:GLR_norm}
\end{figure}

\begin{figure}[hbtp]
\centering
\includegraphics[width=\linewidth]{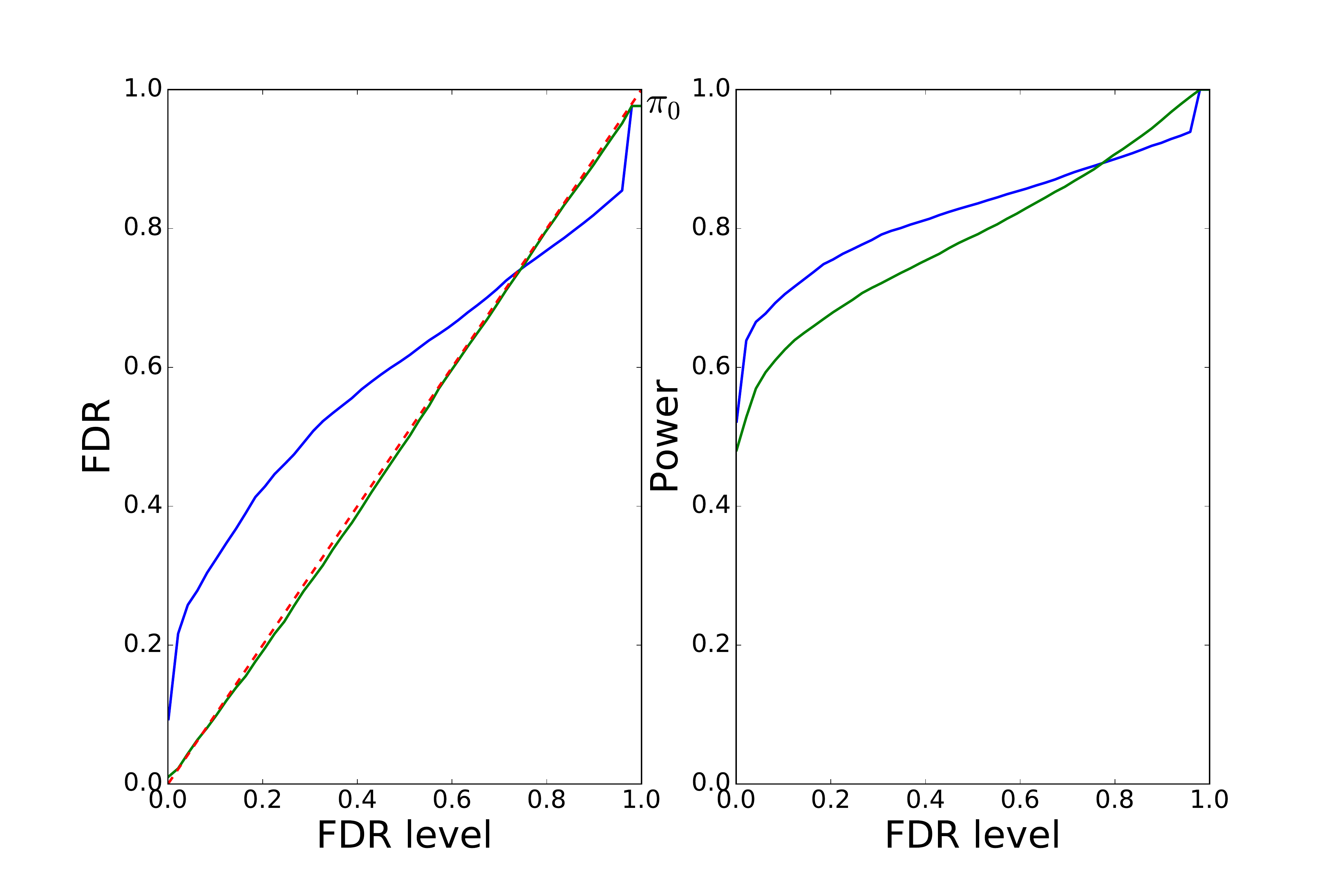}
\caption{True FDR (left) and power (right) versus nominal FDR control level, between GLR and proposed method, on synthesized data with Student noise (4 degrees of freedom). The GLR was calibrated by $10^4$ Monte-Carlo runs on normal noise. $\bSig$ was estimated on the data. Results were averaged on 200 runs of simulated data cubes with $\pi_0=0.97$. GLR is in blue, proposed method in green, $y=x$ in dashed red. Matched filter similarity measure (4) was used for the proposed method.}
\label{fig:GLR_stu}
\end{figure}

\section{Application to the MUSE data}
\label{sec:appli}
It is believed that young galaxies are often surrounded by halos of hydrogen gas, known as the circum galactic medium. The emissions from these halos can be several orders of magnitude fainter than those from the galaxies. Furthermore, the emission spectrum of the halos are composed of narrow lines, notably with the Lyman-$\alpha$ line. 
The MUSE \cite{bacon2010muse} was developed to detect such emissions.
 It is a 3D spectrograph that can image and analyze a field of 1 arcmin$^2$ by producing a hyperspectral data cube of 300 by 300 pixels by 3600 wavelengths. Its spectral range covers the visible and near-infrared domain, from 450 nm to 930 nm.

Since the MUSE first light in January 2014, several studies \cite{paris2013detection,meillier2016} have already been conducted on the MUSE data. The aim has been to detect faint young galaxies, which are also characterized by the presence of a powerful Lyman-$\alpha$ emission line in their emission spectrum. 
Both methods mostly assume spatially and spectrally punctual sources. As a consequence, they efficiently find the core of galaxies, but are not (yet) adapted to detect the faint extended halos. 

The purpose here is to explore the vicinity of these already detected galaxies, and track the Lyman-$\alpha$ emission line as far away from the galaxy as possible. 

\subsection{The MUSE data}
The proposed detection method is applied to the MUSE observations of the sky region known as Hubble Deep Field South (HDFS), as it was previously observed by the Hubble space telescope.
The MUSE produces huge amounts of data that have to be processed by a data reduction system before they can be used for scientific analysis. In particular, a resampling process creates local correlations in all directions (spatial and spectral) between voxels that cannot be easily modelled due to the data dimensions.
The data reduction system applied to the data used here is detailed in \cite{bacon2015muse}. The output is a 300 by 300 pixels by 3600 wavelengths data cube that is associated with a variance cube of the same dimensions. This latter is estimated by propagating the error estimated at the captor level at each stage of the processing.

A catalog of astronomical objects in HDFS was built in \cite{bacon2015muse}. About 90 of these objects are remote galaxies known as Lyman $\alpha$ emitters that are likely to have a halo. 
For each of these sources, a spatial-spectral neighborhood is defined, which is centered spatially on the galaxy center and spectrally on the emission line peak, and a subcube of 50 by 50 pixels by 30 wavelengths is extracted. This cube extraction is performed for the two following reasons.
Spectrally: the signal of interest (hydrogen Lyman $\alpha$ emission line) is concentrated in a few wavelengths around the emission peak. Outside of this domain, galaxy spectra (used to built reference spectrum) can contain other features that are not present in the targeted hydrogen surrounding halo. 
Spatially: as the targeted halo is expected to stay close to the galaxy, exploring empty (only noise) remote regions would only result in a loss in power (in our global control procedure).

\subsection{Pre-processing workflow}
To deal with the MUSE data, several pre-processing steps are needed. First, the spectral continuum is robustly estimated and removed in each pixel using \cite{bacher2016robust}. Then coarse reduction of the data is carried out using the variance cube provided with the data, followed by robust centering and finer reduction, slice by slice \cite{meillier2015detection}.
After the reduction by the variance cube, we can make the following assumption of stationarity of the noise.
\begin{hypothesis}[Stationary noise]
The noise is stationary on each wavelength-slice of the cube.
\end{hypothesis}
~~\\
\subsubsection*{Spatial matched filter preprocessing}
~~\\
For ground-based astronomical instruments such as MUSE, the spatial system impulse response (FSF) is mainly due to atmospheric turbulence. This FSF is measured for each observation (see \cite{villeneuve2011psf}), and is independent of the instrumental noise. As such, we can make the following assumption:
\begin{hypothesis}
The noise in the observation model \eqref{eq:testingPb0} is not filtered by the FSF.
\end{hypothesis}

Based on A7, the following strategy was chosen to improve the SNR: a spatial convolution with the FSF (which is modeled as a symmetrical function) is applied to each image of the wavelength axis of the data cube. It is not strictly speaking a spatial matched filter to the searched halo. Indeed, the halo extension and its intensity profile are not known. 
However, this greatly improves the SNR. The price to be paid is that the theoretical spatial extension of the halo is enlarged by this operation. In practice, the halo 
has a larger extension than the FSF, with an intensity profile that, as does the FSF, decreases quickly toward zero on its support. Therefore,  
this effect can be neglected in the detection results.

\subsection{Detection}
\label{ssec:detection}
In the application on hand, we have the following assumptions:
\begin{enumerate}
\item The galaxy spatial center is already known, as well as the spectral position of the emission line in the galaxy spectrum (with e.g., the method developed in \cite{meillier2016});
\item The emission line in the halo spectra has a shape similar to the emission line in the galaxy spectrum, the continuum of which has been subtracted,  but can present a shift along the spectral wavelengths;
\item Samples (pixels) are weakly dependent.
\end{enumerate}

The second hypothesis is only partially true: in reality the redshift is not a simple spectral shift, but the composition of a shift and a dilatation. However, at the spectral resolution, the deformation can be neglected.
The third assumption is justified because dependences between pixels are mainly due to interpolations during the resampling step, and these dependences have short-range effects. Moreover, in the pre-processing workflow, only the convolution by the FSF creates significant spatial dependences, once again with a short-range effect due to the finite support of the FSF.
Thus, the weak dependence assumption A5 is fulfilled. 
Assumptions required to apply the proposed method are also fulfilled. MUSE data result from the summing of a high number of exposures thus the noise tends to be Gaussian (and as such symmetrical) by application of the central limit theorem. The following numerous preprocessing steps (background subtraction, centering, variance reduction,...) all keep the symmetrical property of the centered noise thus ensuring that assumption A1 hold. Assumption A2 (odd similarity measure) is fulfilled by construction as we choose the SAD measure. As pointed out in \textit{Remark 2} of II-C, assumptions A3 and A4 can't be strictly guaranteed. However, outside of this ideal framework, the key point is that equation \eqref{eq:expH0} is a good enough approximation. Indeed, the targeted signal is assumed to be distinct enough from the background noise and well approximated by the dictionary (see section III-C1); thus $T_{\textrm{max}}(\by)$  will be significantly stochastically larger under $\mathcal{H}_{1}$ than under $\mathcal{H}_{0}$ for detectable signals. Moreover the galaxy and halo pixels are supposed to be in strong minority in the spatially explored region, that is $\pi_0$ is close to one, which enforces these approximations. Finally, as stressed in  \textit{Remark 2}, the approximation errors in equation \eqref{eq:expH0} can only lead to a small loss in power and the control is still guaranteed (the bias is conservative as shown in figure \ref{fig:fdr_control} for low SNR).

For a given galaxy, a spatial-spectral neighborhood is defined, centered spatially on the galaxy center and spectrally on the emission line peak.
Based on these hypotheses, the approach developed here consists of applying the following steps.
\begin{itemize}
\item Estimate a reference emission line spectrum by averaging the spectra of a few pixels at the center of the galaxy.
\item Create a (highly) coherent family of shifted versions of this reference target signature to build a dictionary.
\item Test each pixel of the defined neighborhood using the method developed in section \ref{sec:method}.
\end{itemize}

\subsubsection{Dictionary}
\label{sec:dico}
One of the main assumptions here is that the target signature variability can mostly be modeled as a spectral shift. 
Thus, the dictionary is built here by creating shifted variants of one target signature, $\bd_*$. Assuming that
$\bd_*$ comes from sampling of a continuous model $f(\cdot)$, we can define as $\bd_*^{\delta}$, the shifted
vector that is obtained by sampling $f(\cdot-\delta)$. The {\em linearly spaced shifts} (LSS) dictionary model on an interval $[-\tau,\tau]$
is then defined for a given size $m$, as the dictionary $\bD^m$ composed of the atoms 
$\bd_k= \bd_*^{\tau_k}$, where $\tau_k=-\tau + \frac{2 \tau }{m-1} k$, for $k=0,\ldots,m-1$. 

The key question is then the choice of the number $m$ of shifted versions, or in other words, the redundancy of 
the dictionary. To allow a study of this parameter, we place ourselves in a simplified context:
\begin{itemize}
\item the noise is supposed to be i.i.d.  $\mathcal{N}(0,1)$;
\item the similarity measure is a spectral matched filter between a dictionary atom and the tested spectrum, as in \eqref{eq:MF};
\item the reference spectrum $\bd_*$ is a non-negative vector with unit length, where its autocorrelation function $\Gamma(u)= \langle \bd_*, \bd_*^{u}\rangle$ 
is non-increasing in $|u|$, and has compact support such that $||\bd_*^{u}||=||\bd_*||=1$, for $u\in [-\tau,\tau]$;
\item the target signature $\bx$ is built from a translation $\bd^u_0$ of the reference spectrum $\bx = a\bd^u_0$,  where $a>0$, and $u$ is a random shift that is uniformly distributed on $[-\tau,\tau]$.
\end{itemize}
A measure of the redundancy of a given normalized dictionary $\bD$ can be given by its coherence, which is defined as 
$\mu= \max\limits_{i\neq j}{|\langle\bd_i,\bd_j\rangle|}$.
For a LSS dictionary $\bD^m$, and under the aforementioned assumptions, this coherence reduces to
the correlation between two consecutive atoms:
$\mu= \langle\bd_j,\bd_{j+1}\rangle$, for $1\le j < m$. 
As illustrated by figure \ref{fig:atoms}, by design of the dictionary, the larger the dictionary size $m$, the more correlated the atoms are, and the more coherent the dictionary is.

\begin{figure}[hbtp]
\centering
\includegraphics[width=\linewidth]{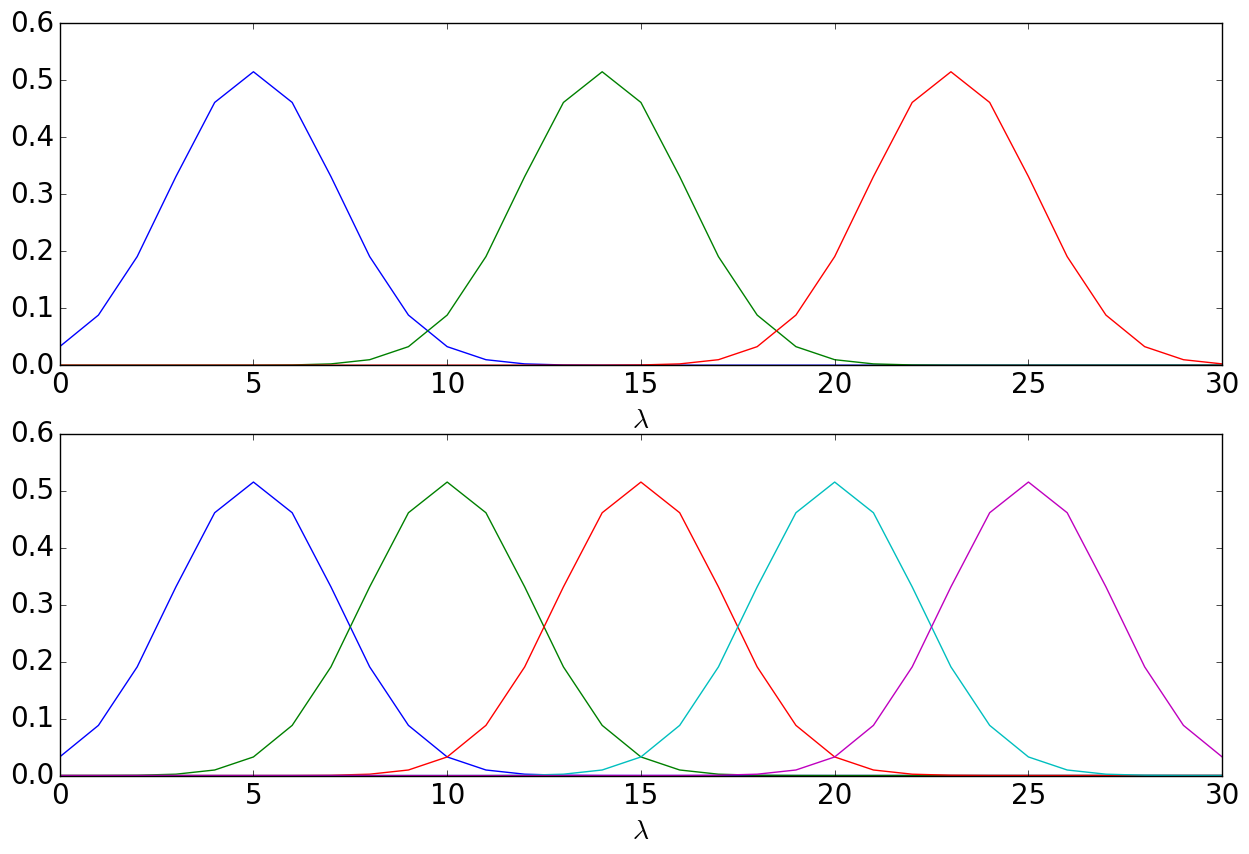}
\caption{Example of dictionaries built from a reference $\bd_*$, with varying number of atoms (with 3 atoms and 0.2 coherence and with 5 atoms and 0.5 coherence). The reference $\bd_* \in \mathbb{R}^l$, with $l=30$,  
is sampled for $j=1,\ldots,l$ from a Gaussian density centered on the median band $j=15$, with full width at half maximum of 5 ($\sigma \approx 2.12$) truncated on $\pm6$ around the mode, and $\ell_2$-normalized. The maximal shift is $\tau=8$.}
\label{fig:atoms}
\end{figure}

Let $\bz^m = (\bD^m)^T \by \in \mathbb{R}^m$ be the vector of the matched filter statistics, the elements of which are defined as $z_j^m$ for $1 \le j \le m$.
For a given decision threshold $\eta$, the PFA for the max-test approach is expressed as 
\begin{align}
 \alpha_m= \Pr{\left(\max{\bz^{m}} > \eta \right)}, \qquad \textrm{under } \mathcal{H}_0.
 \label{eq:defPFA}
\end{align}
Here the noise vector $\beps$ is $\mathcal{N}(0,\bI_m)$ distributed under $\mathcal{H}_0$. If the atoms are orthogonal (e.g., if they have disjoint support),
the vector $\bz^m$ is then normally distributed with zero mean and covariance matrix $\bD^m (\bD^m)^T= \bI_m$. In this case, 
we can compute exactly the PFA as
\begin{align}
\begin{split}
\alpha_m&= 1- \Pr{\left(\max{\bz^{m}} \leq \eta\right)} = 1- \Pr{(z^{m}_1 \le \eta)}^m,\\
&=1- \Phi\left(\eta\right)^m,
\end{split}
\label{eq:exactPFA}
\end{align}
where $\Phi$ is the cumulative function of the normal distribution.
In practice, the dictionary is chosen to be highly coherent (as we want to track close translated versions of a reference spectrum). 
This requires another way to be found to estimate or maximize this probability.
\begin{proposition}
For any $t \in \mathbb{R}$ and $m\ge 2$, let $M_{m+1}(t)$ be defined recursively, under $\mathcal{H}_0$, as
\begin{align}
M_{m+1}(t) = \Pr{\left( z^{m+1}_1 \le t \mid z^{m+1}_2 \le t, z^{m+1}_3 \le t \right)} \times M_{m}(t),
\label{eq:boundPFA}
\end{align}
with $M_2(t)= \Pr{(z^{2}_1 \le t,z^{2}_2 \le t)}$.
Under the aforementioned assumptions, an upper boundary of the PFA $\alpha_m$ is given by  $1-M_m(\eta)$.
 \label{prop:pfa}
\end{proposition}
\begin{proof}
 See Appendix  \ref{sec:append1}.
\end{proof}
The interest of expression \eqref{eq:boundPFA} is that the first factor of the right hand side and the initial value $M_2(t)$ can be evaluated numerically based on quadrature rules for trivariate and bivariate normal distribution functions \cite{genz2009computation}, without the requirement for any Monte-Carlo approximation. Thus, this upper boundary can be easily and accurately computed. When the atoms are uncorrelated, this boundary
is sharp and reduces to \eqref{eq:exactPFA}. Moreover this allows appreciation of, for a given threshold $\eta$, the increase in 
the PFA $\alpha_m$ as a function of the dictionary size $m$ for highly correlated atoms. In a reciprocal way, this allows the evaluation of the threshold $\eta_m$, 
which ensures a false alarm rate that is lower than a given $\alpha$ for any $m\ge 1$.

We can now estimate roughly the potential detection gain under  $\mathcal{H}_1$ as a function of $m$, as follows:
Under $\mathcal{H}_1$, we have assumed that $$\by = a\bd^u_* + \beps$$  with a shift $u \sim \mathcal{U}([-\tau,\tau])$. 
Then, if we assume that the maximum is obtained for the closest atom, which is by assumption the more correlated with $\bd^u_*$, 
the expected  max-test statistic can be approximated by $$E[\max \bz^m] \approx aE[\Gamma(e_m)]$$ where $\Gamma(\cdot)$ is the autocorrelation function of $\bd_*$ and $e_m \sim \mathcal{U}([0,\tau/(m-1)])$ is the shift between $\bd^u_*$ and the closest atom.

Using this expected max-test value under $\mathcal{H}_1$ and the upper boundary on the false alarm given in proposition \ref{prop:pfa}, we can see in Figure \ref{fig:comparaisonPFA_Power} that  when the dictionary size $m$ increases, the max-test statistic can still increase under $\mathcal{H}_1$. 
However, for fixed level control $\alpha$, the test threshold (upper boundary) $\eta_m$ does not increase significantly above a certain size, e.g., $m \ge 10$ in Figure \ref{fig:comparaisonPFA_Power}. This is clearly explained by the stronger correlations when adding new atoms
Conversely, if the atoms are uncorrelated, we see that the threshold deduced from \eqref{eq:exactPFA} increases faster than the potential gain of the max-test statistic under $\mathcal{H}_1$.

This is confirmed in Figure \ref{fig:ROC}, which shows different empirical ROC curves for different sized $m$ of the LSS dictionaries. It can be seen empirically that the more coherent the dictionary, the more powerful becomes the max test.
\begin{figure}[hbtp]
\centering
\includegraphics[width=0.8\linewidth]{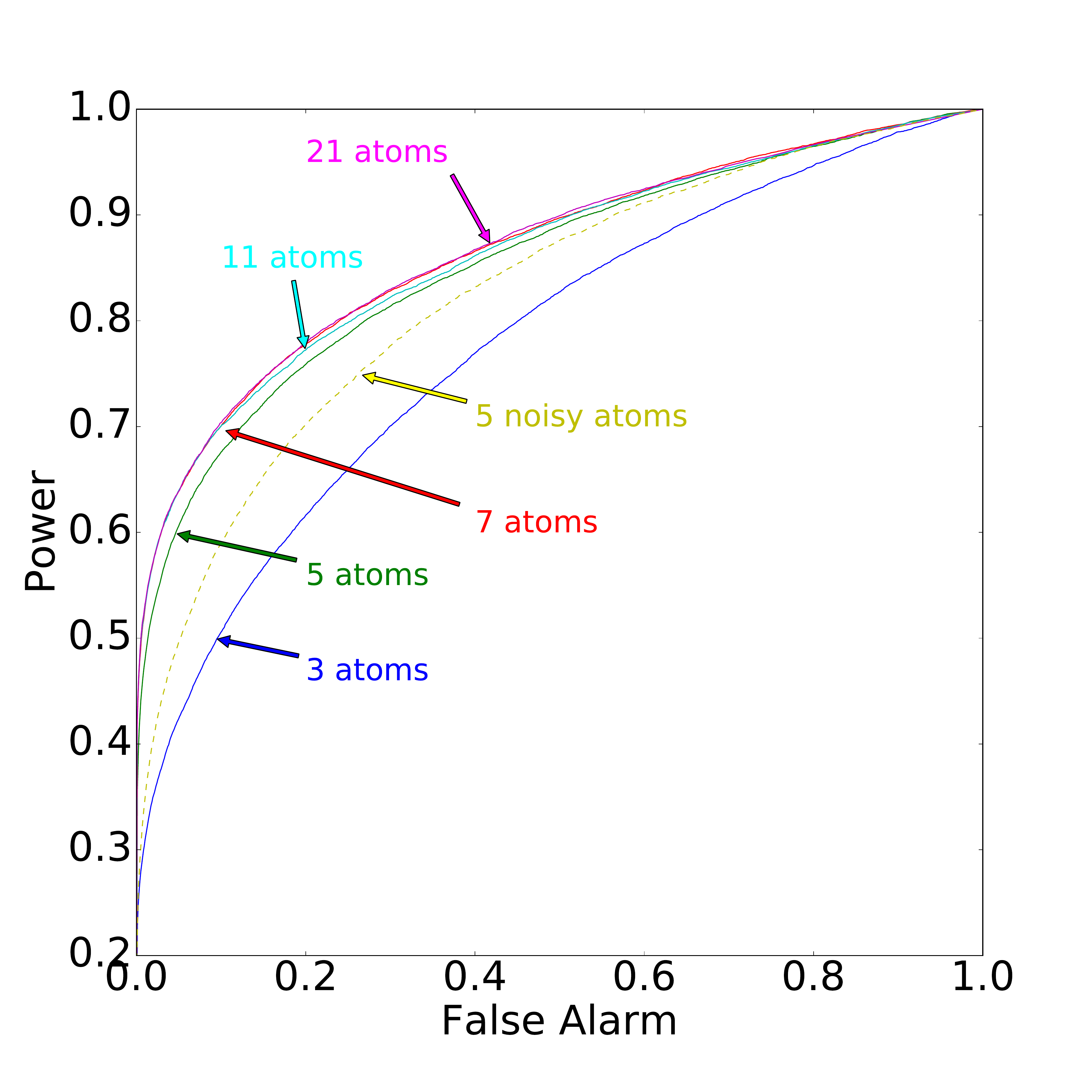}
\caption{Comparison of empirical ROC curves for several dictionary sizes $m$ under the LSS dictionary model.
The results were obtained on $50$ Monte-Carlo runs of simulated MUSE-like data under the aforementioned assumptions.}
\label{fig:ROC}
\end{figure}

As a consequence, in the application, the dictionary will be built to be as coherent as possible for the spectral resolution of the MUSE instrument. The reference atom $\bd_{*}$ is estimated by averaging the spectra of the 5 pixels at the spatial intensity peak of the galaxy. The spectrum is limited to a $l=30$ spectral band area centered on the spectral emission peak, which ensures the presence of the whole emission line feature. Based on astronomical priors, the spectral shift is limited to the interval $[-\tau,\tau]$ with $\tau=7$ MUSE spectral  bands (i.e., $\tau\approx \SI{9}{\angstrom}$). Shifting is done at the spectral resolution of the instrument, to avoid any interpolations. The dictionary $\bD^m$ is finally built with the atoms corresponding to these $m=15$ shifted versions of $\bd_{*}$.

\begin{figure}[hbtp]
\centering
\includegraphics[width=0.8\linewidth]{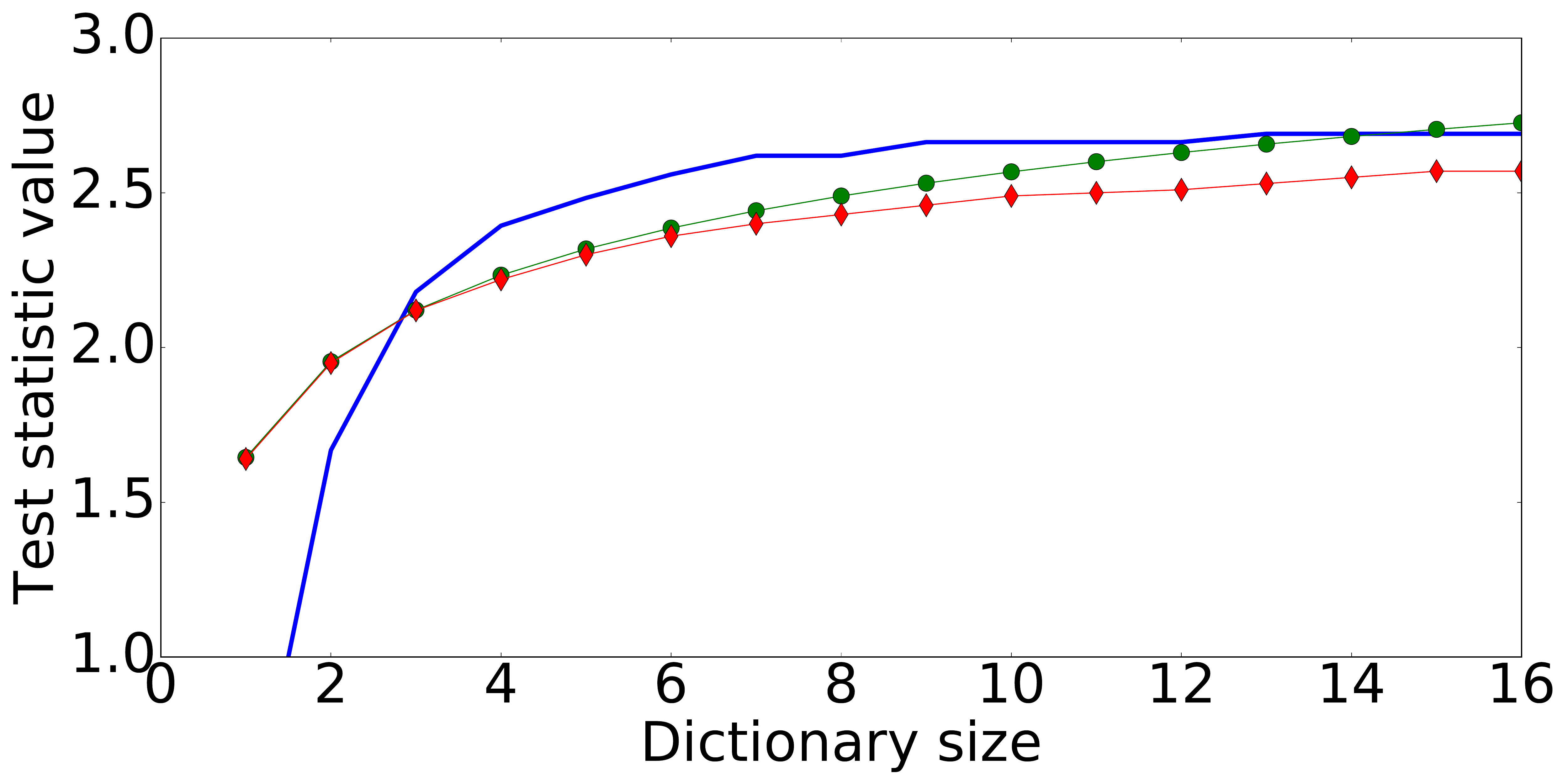}
\caption{Test statistics threshold versus the number of dictionary atoms $m$ for a control level of $\alpha=0.05$.  $\begingroup\color{red}\blacklozenge\endgroup$ marker curve, threshold $\eta_m$ 
for the LSS dictionary $\bD^m$ as in Fig. \ref{fig:atoms};
   $\begingroup\color{OliveGreen}\bullet\endgroup$
marker curve, threshold for a size $m$ dictionary 
with uncorrelated atoms; heavy blue curve, evolution of the test statistic potential gain  under $\mathcal{H}_1$ for an intensity $a=2.7$.
\label{fig:comparaisonPFA_Power}}
\end{figure}

\subsubsection{Similarity measure}
To test whether a given spectrum belongs to the extended source, the similarity measure used in this application is the 
SAD, as defined in \eqref{eq:SAD}. Of note, other metrics were explored to build the test statistic,
such as the matched filter one \eqref{eq:MF} and the spectral information divergence defined in \cite{chang1999spectral}.
Spectral information divergence is built upon the symmetrical Kullback-Leibler divergence, and it compares the spectra as distribution densities. As it demands positive signals for its computation, 
it cannot be used directly for our problem, as the MUSE data can be negative due to high symmetrical noise levels.  
The matched filter approach can be used on the MUSE data, and it gives good results. However SAD appears to be 
more robust to some systematics of the MUSE data cubes, such as the edges where there is higher variability, and it is preferred here.

\subsection{Results on real data}
\label{ssec:results}
The results on several subcubes of $n=50 \times 50$ pixels by $l=30$ wavelengths centered around interesting objects of the HDFS catalog are shown in Figure \ref{fig:result} for the detection procedure described in Alg. \ref{algo:FDR}. For each of the $n=50\times50=2500$ spectra, the max-test statistics \eqref{eq:maxtest} are obtained from the SAD similarity measure, and with a highly coherent dictionary constructed as described in the last paragraph of section \ref{sec:dico}. The empirical null estimators are computed on larger subcubes (centered around the subcube to be tested) that are composed of  200 by 200 spectra. 
For each object, the first row shows the narrow band image around the emission line (the data subcube is totalled along the $l=30$ spectral bands centered on the emission line peak). The second column of the first row shows the same narrow-band image, but after the different preprocessing steps (which include continuum subtraction and FSF convolution). The last column shows the reference spectrum $\bd_*$, built from the pixels at the center of the studied object.
On the second row, the first column shows the maps of the empirical 
$p$-values \eqref{eq:pvalemp} obtained for the $n=2500$ spectra. The maps of the  $q$-values are depicted on the second column.
Q-values were introduced in \cite{storey2003statistical} and can be seen as the FDR counterpart of the $p$-values. 
For each  test statistic, it is defined as the minimal FDR that allows this test statistic to be a discovery. In our detection framework, this is the 
minimal FDR control level $q\in [0,1]$ in Alg. \ref{algo:FDR}, such that a given spectrum is detected as part of the halo.
The interest in this global measure of significance is clear here, as it allows us to present more contrasted significance maps than the classical $p$-value maps. 
The third column shows the binary detection map provided by  Alg. \ref{algo:FDR} for a nominal FDR control level $q=0.2$. The contours of the detection region for different values of $q$ are also superimposed. 
From these maps, it can be seen that several of these objects show clear asymmetry, and that they extend beyond the simple support of a punctual source (the black circles show the support of an estimated FSF). Studies are currently being conducted by astronomers at the Centre de Recherche Astrophysique de Lyon to analyze these results and to apply the method to other sky fields.

\begin{figure}[hbtp]
\centering
\subfloat[Object 43]{
\begin{minipage}{\linewidth}
\centering
\includegraphics[width=0.9\linewidth]{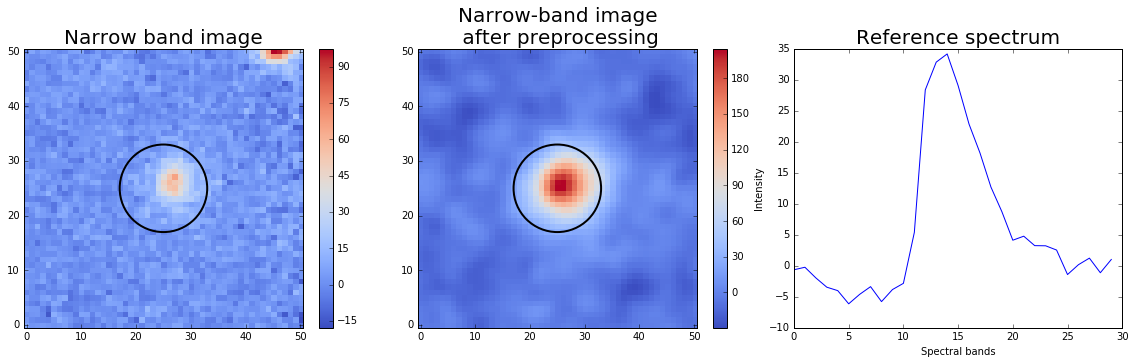}
\includegraphics[width=0.9\linewidth]{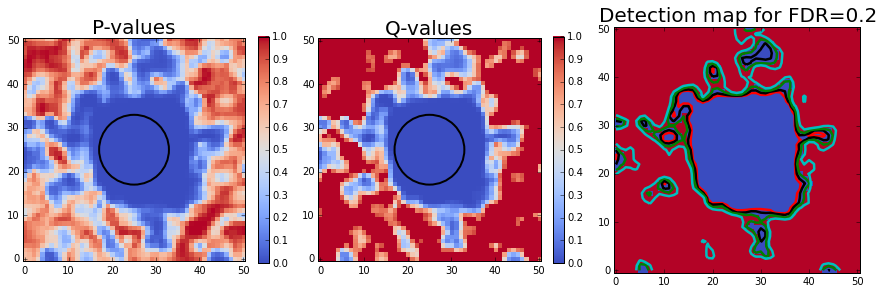}
\end{minipage}%
}\\
\subfloat[Object 92]{
\begin{minipage}{\linewidth}
\centering
\includegraphics[width=0.9\linewidth]{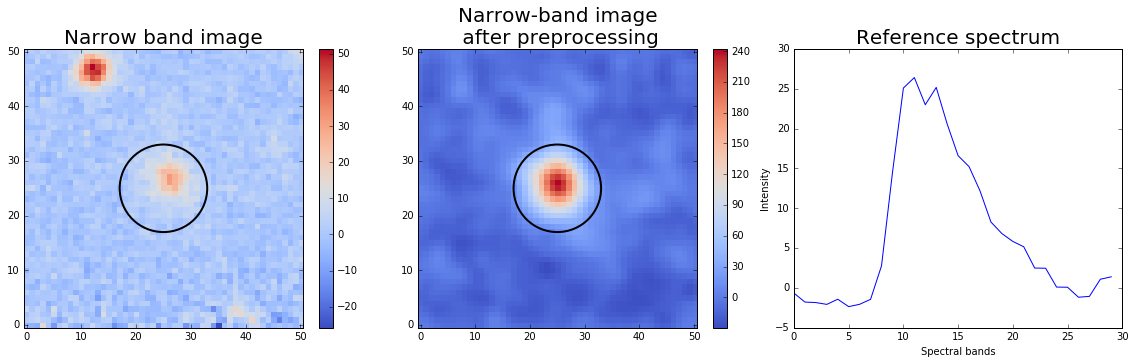}
\includegraphics[width=0.9\linewidth]{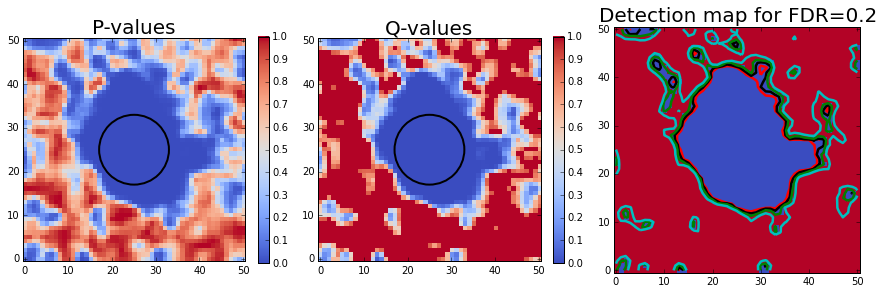}
\end{minipage}%
}\\
\subfloat[Object 139]{
\begin{minipage}{\linewidth}
\centering
\includegraphics[width=0.9\linewidth]{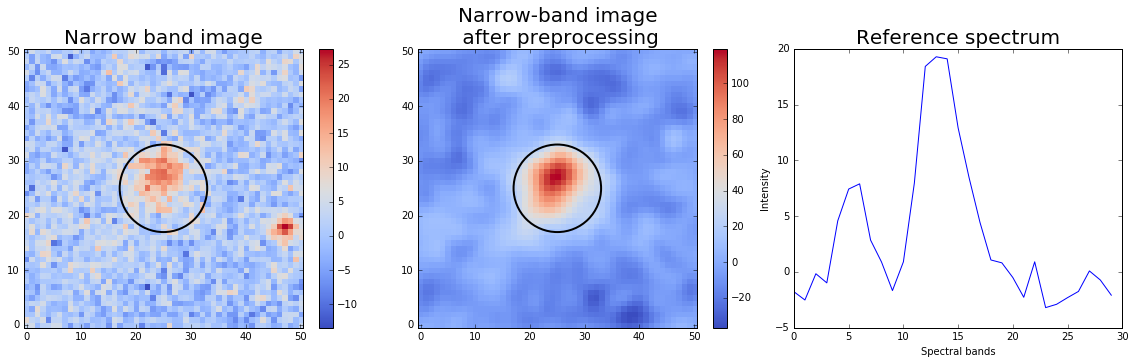}
\includegraphics[width=0.9\linewidth]{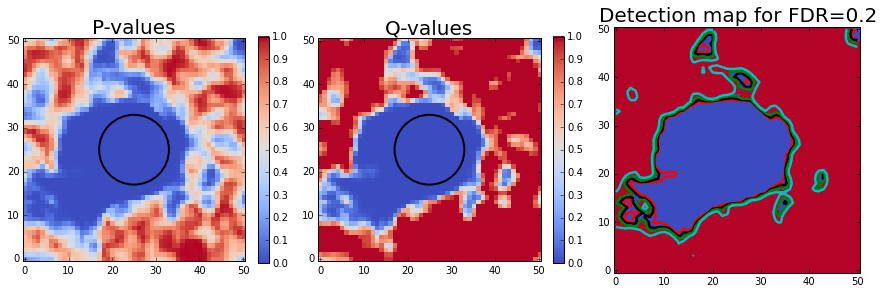}
\end{minipage}%
}
\caption{Results on several objects of the HDFS. Each subfigure corresponds to an object with the numbering to the catalog defined in \protect\cite{bacon2015muse}. For each subfigure: on the first row, from left to right: narrow-band image, narrow-band image after preprocessing (continuum subtraction and FSF convolution), and reference spectrum atom.
On the second row, from left to right: p-value map resulting from the test, q-value map, and detection contours for several FDR levels $q$ (0.05 in red, 0.1 in black, 0.2 in green, and 0.4 in cyan) superimposed on the binary detection map for a FDR level $q=0.2$ (blue pixels, halo detections). The circles in black in the first two columns show the extent of the FSF. Results are analyzed in \ref{ssec:results}.}
\label{fig:result}
\end{figure}

\section{Conclusions}
\label{sec:conclu}
In this paper, a new method is proposed to answer a detection problem of a weak target signature that is partially known, but with a possible large variability within an unknown background that is difficult to model.
To answer to this problem, an unsupervised detector was proposed, based on a maximum test approach, as studied in \cite{arias2011global}. This detector takes explicitly the possible target variability into account by using a highly coherent dictionary. It does not need any knowledge of the background, but a simple noise symmetry assumption, and the non-negativity of the sparse representation of the targeted signal. This allows to estimate the test statistic distribution and to implement a simple detection procedure robust to model/background miss-specification. Moreover, the error control was developed based on a false discovery rate approach, and a global measure of the significance was obtained. Such a control with detection threshold that adapts to the data is not yet widely used in the signal processing community whereas it is highly pertinent for processing massive datasets. 
This whole new process was tested on real MUSE data. The promising obtained results are presently analyzed by astronomers.
Future extensions of this original method will account for the existence of spatial structure of the target while controlling the FDR.
 The MUSE data used in this paper is now publicly available at \url{http://muse-vlt.eu/science/hdfs-v1-0/}, and the Python code of the proposed method is available on demand.

\section*{Acknowledgments}

The authors would like to thank the ERC 339659-MUSICOS for its funding, Roland Bacon and Floriane Leclercq for their expertise on the MUSE data, and Jean-Baptiste Courbot for numerous interesting exchanges on weak target detection.

\appendix
\makeatletter
\def\@seccntformat#1{Appendix~\csname the#1\endcsname:\quad}
\makeatother

\subsection{Proof for the empirical null estimators}

\subsubsection{Proof of lemma \ref{le:med}}
\label{app:le}

Let $\bg=\left(T_{\textrm{max}}(\by_1),\ldots,T_{\textrm{max}}(\by_n)\right)$ be the set composed of the $n$ max statistics, the elements of which are
denoted as $g_i$ for $1\le i \le n$
Similarly
$\bs=\left(-T_{\textrm{min}}(\by_1), \ldots,-T_{\textrm{min}}(\by_n)\right)$ is the set composed of the $n$ opposite min statistics, the 
elements of which are denoted as $s_i$ for $1\le i \le n$.

We first show that $\widehat{\mu}_0$  verifies \eqref{eq:crossmed}, i.e., that  $\#\{g_i \le \widehat{\mu}_0\} = \#\{s_i > \widehat{\mu}_0\}$.
For absolutely continuous distributions, $\Pr( t_{(n)}= t_{(n+1)})=0$.
Thus from \eqref{eq:mu0est}, we get that  $\#\{ t_i \le \widehat{\mu}_0\}= n$ with probability one.
The sample set $\bt$ is the union of $\bg$ and $\bs$: if $m_0= \#\{ g_i \le \widehat{\mu}_0\}$, then  
 $\# \{ s_i \le \widehat{\mu}_0\}=n-m_0$. As a consequence, $\# \{ s_i > \widehat{\mu}_0\}=n-(n-m_0)= m_0= \#\{ g_i \le \widehat{\mu}_0\}$, which shows that 
 $\widehat{\mu}_0$  verifies \eqref{eq:crossmed}.
 
We show now that $\widehat{\mu}_0$ converges in probability toward $\mu_0$. As  $\widehat{\mu}_0$ satisfies \eqref{eq:crossmed}, and 
$\overline{F}(t)$ (resp. $\overline{G}(t)$) converges in probability to $F(t)$ (resp.  $G(t)$) for any $t \in \mathbb{R}$, $\widehat{\mu}_0$
converges in probability to the solution of
\begin{align*}
 F(t)= G(t).
\end{align*}
if this equation admits a unique solution.
Assuming that the median of $F_0$ is uniquely defined, it follows that for $t>\mu_0$ 
\begin{align*}
F(t)=\pi_0 F_0(t) +  \pi_1 F_1(t) \ge \pi_0 F_0(t) > \pi_0 F_0(\mu_0)= \pi_0/2,
\end{align*}
Moreover, for $t>\mu_0$ 
\begin{align*}
G(t)=\pi_0 G_0(t)   < \pi_0 G_0(\mu_0)= \pi_0/2 < F(t),
\end{align*}
where the first equality is due to zero assumption A4. 
As a consequence, there is no solution of $(\mu_0,+\infty)$.
Similarly, according to zero assumption A3, that 
for $t<\mu_0$,  $F(t) < G(t)$. 
Therefore the unique solution is for $t=\mu_0$, where 
$F(\mu_0)=\pi_0 F_0(\mu_0)=\pi_0/2= \pi_0 G_0(\mu_0) =  G(\mu_0)$, 
which concludes the proof. \hfill $\square$

\medskip
\subsubsection{Proof of proposition \ref{prop:nullest}}
\label{app:prop}
We first show that the $\pi_0$ estimator given in \eqref{eq:pi0} is consistent.
From lemma \ref{le:med}, $\widehat{\mu}_0$ converges in probability toward $\mu_0$: 
$\widehat{\mu}_0 \overset{P}{\longrightarrow} \mu_0$.
The triangular inequality ensures that $|\overline{F}(\widehat{\mu}_0)- F(\mu_0)| \le |\overline{F}(\widehat{\mu}_0) - F(\widehat{\mu}_0)| + |F(\widehat{\mu}_0) - F(\mu_0)|$.
The first term of the right-hand side is dominated by $\sup_t |\overline{F}(t) - F(t)|$,, which converges in probability toward $0$, according to assumption A5. 
The second term also converges in probability toward $0$, according to the continuous mapping theorem. Thus 
$\overline{F}(\widehat{\mu}_0) \overset{P}{\longrightarrow} F(\mu_0)$. 

According to \eqref{eq:expH0}, $F(\mu_0)= \pi_0 F_0(\mu_0)= \frac{\pi_0}{2}$.
As $2 \overline{F}\left(\widehat{\mu}_0\right)= 2\frac{n_0}{n}$, this shows that 
\begin{align*}
\widetilde{\pi}_0 \equiv 2\frac{n_0}{n} \overset{P}{\longrightarrow} \pi_0 \in (0,1].
\end{align*}
Thus $\widehat{\pi}_0=\min\left\{\widetilde{\pi}_0, 1 \right\}$ also converges in probability to $\pi_0$.

We show now the consistency of \eqref{eq:F0empest} for $t \in \mathbb{R}$.
From \eqref{eq:expH0} and assumption A5, it follows now that $\overline{F}(t) \overset{P}{\longrightarrow} \pi_0 F_0(t)$ 
for all $t \le \mu_0$. Then, according to the Slutsky theorem, $\overline{F}(t) /\widetilde{\pi}_0 \overset{P}{\longrightarrow}  F_0(t)$. As $\widehat{F}_0(t)=\overline{F}(t) /\widetilde{\pi}_0$ for all $t \le \mu_0$,   this  
shows the consistency for  $t \le \mu_0$.
The demonstration for $t > \mu_0$ can be done in  a similar manner,  by noting that 
$\widehat{F}_0(t)= 1- \overline{G}(t)/ \widetilde{\pi}_0$ for $t > \mu_0$. \hfill $\square$

\medskip
\subsection{Proof of proposition \ref{prop:storey}}
\label{app:storey}
Let $T_{\textrm{max}}(\by_{(1)}) \leq T_{\textrm{max}}(\by_{(2)}) \leq \cdots  \leq T_{\textrm{max}}(\by_{(n)})$ be the ordered max-test statistics, while $p_{(1)} \leq p_{(2)} \leq \cdots  \leq p_{(n)}$ are the ordered $p$-values (while $p_{i}$ notes the p-value associated with pixel $i$).
From \eqref{eq:pvalemp}, it follows that, $$p_{(i)} = 1 - \widehat{F}_0 \left( T_{\textrm{max}}(\by_{(j)}) \right) \quad \textrm{for } 1\le i \le n$$ where $j=n-i+1$. 
For $j \leq n_0$, $T_{\textrm{max}}(\by_{(j)}) \leq \mu_0$ thus $\# \left\{ s_{0,i} \le T_{\textrm{max}}(\by_{(j)}) \right\} = j$ and $\# \left\{ g_{0,i} \le T_{\textrm{max}}(\by_{(j)}) \right\} = 0$.
So, $2n_0\widehat{F}_0 \left( T_{\textrm{max}}(\by_{(j)})\right) =j$ for $j \leq n_0$, that is $2n_0p_{(i)} = 2n_0 - n+i-1$ for $i\geq n-n_0+1$.
For $k \geq n_0$, let $i_k=n-2n_0 +1+k$. Then $i_k \geq n-n_0+1$ so $2n_0 p_{(i_k)}=2n_0 - n+i_k-1=k$. Thus $\#\{ 2n_0 p_{i} >k \}=n-i_k=2n_0 - k -1$.
If $\zeta= {k}/{2n_0}$, then $\#\{ p_{i} >\zeta \}= \#\{ 2n_0 p_{i} >k \}$. Thus $\frac{ 1+ \#\{ p_{i} >\zeta \}}{(1-\zeta) n}= \frac{2n_0-k}{(1- k/2n_0)n}= \frac{2n_0}{n}$, which 
 shows that $\hat{\pi}_0^{*}(\zeta) = \widehat{\pi}_0$. 
 
In the general case where $\zeta \in [\frac{1}{2},1)$, 
then $\zeta$ and $k_{\zeta}/2n_0$, where $k_{\zeta}=\left \lfloor{2n_0 \zeta}\right \rfloor \in \{n_0,\ldots,2n_0-1\}$, are asymptotically equivalent when
$n_0$ grows to infinity.  
Thus, $\hat{\pi}_0^{*}(\zeta)$ and $\hat{\pi}_0^{*}(k_{\zeta}/2n_0)= \widehat{\pi}_0$ are
asymptotically equivalent. This concludes the proof.
\hfill $\square$

\subsection{Proof of proposition \ref{prop:pfa}}
\label{sec:append1}
Under $\mathcal{H}_0$, for a threshold $t$ we have:
\begin{align*}
\Pr{\left(\max\bz^{m+1} \le t\right)} &= \Pr{\left(z^{m+1}_1 \le t,...,z^{m+1}_{m+1} \le t\right)}\\
					&= \Pr{\left(z^{m+1}_1 \le t \mid z^{m+1}_2\le t ,...,z^{m+1}_{m+1} \le t\right)}\\
					& \times \Pr{\left(z^{m+1}_2\le t,..., z^{m+1}_{m+1}\le t\right)}
\end{align*}
As $\bD^m \geq 0$  and $\by \sim \mathcal{N}(0,\bI_m)$ under $\mathcal{H}_0$, $\bz^m$ is positively associated in the sense of \cite{esary1967association}.
Thus,
\begin{align}
\begin{split}
\label{eq:conditional}
\Pr{(z^{m+1}_1 \le t \mid z^{m+1}_2\le t ,...,z^{m+1}_{m+1}\le t)} \geq \\ \Pr{(z^{m+1}_1 \le t \mid z^{m+1}_2\le t,z^{m+1}_3\le t)}
\end{split}
\end{align}
Using the numerical procedures given in \cite{genz2009computation}, we can accurately compute the right-hand side term of \eqref{eq:conditional}.
Note that this term gives a relatively sharp lower boundary, because $z^{m+1}_2$ and $z^{m+1}_3$ are the more correlated variables with $z^{m+1}_1$ among the $z^{m+1}_j$ for $j\ge 2$.

Moreover, by construction, the shifts between the atoms in $\bD^{m+1}$ are smaller than those for the atoms in $\bD^{m}$. With the autocorrelation function  assumed to be non-increasing with the absolute shifts, it follows that the size $m$ Gaussian random vector $\left(z^{m+1}_2 ,...,z^{m+1}_{m+1} \right)$
has larger correlations than the size $m$ Gaussian random vector $\left(z^{m}_1 ,...,z^{m}_{m}\right)$. By assumption, these two vectors are centered with 
unit marginal variances under $\mathcal{H}_0$. Thus the Slepian lemma \cite{slepian1962one} yields that:
\begin{align} \hspace{-2.5mm}
\label{eq:slepian}
\Pr{\left(z^{m+1}_2\le t ,...,z^{m+1}_{m+1} \le t \right)} \geq \Pr{\left(z^{m}_1\le t ,...,z^{m}_{m} \le t\right)}.
\end{align}
By combining \eqref{eq:conditional} and \eqref{eq:slepian}, we can then minimise  $\Pr{(\max\bz^{m} \le t)}$ by a function $M_m(t)$ that is defined recursively as:
$$M_{m+1}(t) = \Pr{\left(z^{m+1}_1 \le t \mid z^{m+1}_2\le t ,z^{m+1}_3\le t \right)} \times M_{m}(t),$$
where $M_2(t)= \Pr{(z^{2}_1 \le t,z^{2}_2 \le t)}$. This gives the upper boundary for the PFA of proposition \ref{prop:pfa}.
Numerical computations emphasize that $M_m(t)$ increases with $t$. Then it is possible to (numerically) inverse $M_m(\eta)$ to get $\eta_m$ for a control level $\alpha$: $\eta_m=M_m^{-1}(1-\alpha)$ verifies $\Pr(\max\bz^{m} > \eta_m) \leq \alpha$.
\hfill $\square$


\bibliographystyle{IEEEtran}
\bibliography{tsp}

\end{document}